\pdfoutput=1
\documentclass[12pt]{article}
\usepackage{color}
\usepackage[latin1]{inputenc}
\usepackage[T1]{fontenc}
\usepackage[english]{babel}
\usepackage{amsfonts}
\usepackage{amssymb}

\usepackage[margin=2.5cm]{geometry}
\usepackage{titlesec}
\titleformat{\subsection}
  {\large\bfseries\itshape}{\thesubsection.}{1em}{}

\usepackage{graphicx}
\usepackage{dsfont}
\usepackage{epsfig}
\usepackage{epstopdf}
\usepackage{ae}
\usepackage{amsmath}
\usepackage{amsthm}
\usepackage{thmtools}
\usepackage{endnotes}
\usepackage{setspace}
\usepackage{amsthm}
\usepackage{appendix}
\usepackage{float}
\usepackage{subfigure}
\usepackage{multirow}
\usepackage{chngpage} %
\usepackage{enumerate}
\usepackage{enumitem}

\usepackage{url}

\def\dim{d}

\newcommand{\N}{\mathcal{N}}
\newcommand{\E}{E}

\newcommand{\R}{\mathbb{R}}

\renewcommand{\(}{\left(}
\renewcommand{\)}{\right)}

\renewcommand{\geq}{\geqslant}
\renewcommand{\leq}{\leqslant}

\def\R{\mathbb{R}}
\def\Prob{\mathbb{P}}
\def\E{\mathbb{E}}
\DeclareMathOperator{\VaR}{VaR}
\DeclareMathOperator{\TVaR}{TVaR}

\def\Expectile{e}

\newcommand{\abs}[1]{\left|#1\right|}
\DeclareMathOperator*{\argmin}{argmin} %
\DeclareMathOperator{\median}{med}
\newcommand{\ind}[1]{\mathds{1}_{#1}}
\newcommand{\xNorm}[2]{\left\Vert #1 \right\Vert_{#2}}
\newcommand{\scalarProduct}[2]{\left\langle #1 , #2 \right\rangle}
\DeclareMathOperator{\tr}{\top}

\newtheorem{Proposition}{Proposition}[section]
\newtheorem{Theorem}{Theorem}[section]
\newtheorem{Definition}{Definition}[section]
\newtheorem{Lemma}{Lemma}[section]
\newtheorem{Corollary}{Corollary}[section]

\usepackage{natbib}
\usepackage{verbatim}
\usepackage{bm} %
\newcommand{\NN}{\mathbb{N}}

\def\bfa{\bm{a}}
\def\bft{\bm{t}}
\def\bfu{\bm{u}}
\def\bfx{\bm{x}}
\def\bfy{\bm{y}}
\def\bfX{\bm{X}}
\def\bfc{\bm{c}}
\def\bfalpha{\bm{\alpha}}
\def\bfbeta{\bm{\beta}}
\def\bfZ{\bm{Z}}
\def\bfz{\bm{z}}

\def\bfzero{\bm{0}}
\def\bfY{\bm{Y}}
\def\bfxi{\bm{\xi}}
\def\bfU{\bm{U}}
\def\bfE{\bm{E}}

\def\Copula{C}
\DeclareMathOperator{\Poisson}{Pois}
\DeclareMathOperator{\Exponential}{Exp}

\def\inprob{\overset{p}{\longrightarrow}}
\def\as{\overset{a.s.}{\longrightarrow}}
\def\eqdist{\overset{d}{=}}

\def\indI{\ell}
\def\bfmu{\bm{\mu}}

\begin{document}

\author{\textsc{Klaus Herrmann}\footnote{Department of Mathematics
and Statistics, Concordia University, 1400 de Maisonneuve Blvd. West,
Montr\'eal (Qu\'ebec) Canada H3G 1M8; e-mail:
\texttt{klaus.herrmann@concordia.ca} }
\and \textsc{Marius Hofert}\footnote{Department of Statistics and
Actuarial Science, University of Waterloo, 200 University Avenue West,
Waterloo (Ontario), Canada N2L 3G1; e-mail:
\texttt{marius.hofert@uwaterloo.ca}}
\and \textsc{M\'elina Mailhot}\footnote{Department of Mathematics and
Statistics, Concordia University, 1400 de Maisonneuve Blvd. West,
Montr\'eal (Qu\'ebec) Canada H3G 1M8; e-mail:
\texttt{melina.mailhot@concordia.ca}.}
}

\title{\bf\LARGE Multivariate Geometric Expectiles}
\date{\today}

\maketitle

\begin{abstract}
\noindent
A generalization of expectiles for $\dim$-dimensional multivariate distribution functions is introduced.
The resulting geometric expectiles are unique solutions to a convex risk minimization problem and are given by $\dim$-dimensional vectors.
They are well behaved under common data transformations and the corresponding sample version is shown to be a consistent estimator.
We exemplify their usage as risk measures in a number of multivariate settings, highlighting the influence of varying margins and dependence structures.
\end{abstract}


\setcounter{page}{1}

\section{Introduction}
A fundamental task in risk management and applied actuarial science is to quantify the risks associated with a given position.
Prime examples of risky positions are portfolio holdings or (re-)insurance contracts.
Quantifying risks is not only necessary for the internal decision making process of financial institutions, insurance companies or individual investors, but also mandatory from a regulatory perspective.
For example, the regulatory frameworks for banks (OSFI, AMF, Basel II, 2.5, III) and insurance companies (CIA and, in Europe, Solvency II, Swiss Solvency Test) require not only internal risk modeling, but also specifically demand that businesses quantify and report risks in a specific way, using risk measures.
This task is intrinsically multivariate in nature as one of the American Property and Casualty Minimum Capital Target Advisory Committee key principle is that \lq Risks should be aggregated. No diversification between risk categories is permitted until evidence confirms diversification will hold in a stress situation\rq\phantom{a}(see~\cite{OSFI2010}).
The Office of the Superintendent of Financial Institutions of Canada states : \lq Gross, ceded and net provisions for claims liabilities must be provided by actuarial lines of business\rq\phantom{a}(see~\cite{AAR2014}).

Until recently, regulatory economic capital has been calculated based on univariate risk measures.
In this case, i.e., when considering risks separately, the theory of risk measures is well established, see, e.g., \cite{McNeilFreyEmbrechts2015} Chapter 2 for an overview.
The two most popular risk measures in this setting are value-at-risk (VaR) and tail-value-at-risk (TVaR; sometimes also referred to as conditional-tail-expectation or expected shortfall).

However, capital allocation has to be reinvestigated when dealing with a portfolio when it is more appropriate to secure capital simultaneously for multiple business activities.
In this paper, we introduce a method that allows users to allocate capital to each risk based on possibly different confidence levels, and considering the dependence between and among business lines.

In a real world scenario markets and assets are interconnected or prone to systemic risk.
The same holds true when considering insurance contracts where dependence can play an important role.
It can thus be beneficial to consider risks in a joint framework rather than treating them as isolated entities, such as the top-down allocation rule.
Many problems from this consideration have been pointed out in the last decade, notably by Bank of Canada (see~\cite{gauthier2010}).
To this end, a general theory for multivariate risk measures which specifically take the underlying dependence structure into account has recently emerged.

Based on multivariate risk measures, the trade-off between two stock indices has been studied by \cite{Cherubini2004} using bivariate inverse quantiles.
Losses and adjustment loss allocation expenses (ALAE) have been studied by \cite{di2013plug}, using multivariate value-at-risk and tail-value-at-risk.
\cite{GueganHassani2014} allocate risk capital based on bivariate quantiles, where operational risk and other related risks are considered as separate dependent classes.
Most of the techniques use an acceptance set, as presented in \cite{Jouini2004}, and calculate a metric for each risk class, considering the dependence between those classes.
\cite{balbasEtAl2011} present several properties from a general representation of multivariate risk functions.

From an actuarial perspective, multivariate risk measures generalizing VaR are treated in \cite{EmbrechtsPuccetti2006}, \cite{CossetteMailhotMarceau2012}, \cite{CousinDiBernardino2013} and \cite{TorresLilloLaniado2015}.
Multivariate versions of TVaR have been defined in \cite{CousinDiBernardino2014} and \cite{CossetteMailhotMarceauMesfioui2015}.
\cite{MaumeDeschampsRulliereSaid2017} also introduce a multivariate extension of expectiles.
However, this approach differs from ours in the sense that it is non-geometrical.
Likewise, the statistical community has generalized the notion of quantiles, i.e., VaR, to higher dimensions, e.g., via the notion of statistical depth, see, e.g., \cite{Mosler2013} for an overview, and optimization-based definitions as in \cite{AbdousTheodorescu1992}, \cite{Chaudhuri1996} or \cite{Chakraborty2001}.
Although the two approaches set out from different starting points, interconnections are possible in some cases, see for example \cite{HallinPaindaveineSiman2010}.
A thorough overview of different approaches can be found in \cite{Serfling2002}, while a connection between half-space depth and stress testing risk factors is established in \cite{McNeilSmith2012}.

Our work is motivated by the fact that despite its good properties and popularity, the tail-value-at-risk is not elicitable in the univariate case, see \cite{Gneiting2011}.
Elicitability is a property that has been investigated in \cite{osband1985providing} in order to score the estimation of risks.
Therefore, using the same criterion to do forecasting-based model selection and risk assessment is not possible when using $\TVaR$, or, as shown in \cite{Ziegel2014}, any other spectral risk measure other than the expectation.
While univariate quantiles are elicitable, and can thus be utilized in forecasting-based model selection, they, however, do not adhere to the broadly accepted framework of coherence, see \cite{ArtznerDelbaenEberHeath1999}, which establishes preferable properties of risk measures in an axiomatic fashion.
This is a serious drawback in actuarial applications.

As shown in \cite{Ziegel2014}, the only elicitable, law-invariant and coherent risk measures are expectiles, introduced by \cite{NeweyPowell1987}.
Expectiles generalize the mean for a given probability distribution in much the same way as quantiles generalize the median.
Furthermore, they have a natural interpretation when considering the gain-loss ratio connected to a given position, i.e., the ratio of the expected gains over the expected losses, which is a popular performance measure in portfolio management, see \cite{BelliniDiBernardino2017}.
The amount of money that needs to be added to a given position in order to achieve a pre-specified, and in practical applications sufficiently high, gain-loss ratio is given by an expectile.
In the univariate case, expectiles therefore combine favourable properties of risk measures and constitute an important addition to the well established VaR and TVaR.

Considering both, the need for multivariate risk measures and the favourable properties of univariate expectiles, the main target of the present study is therefore to define a multivariate version of expectiles and to study its properties.
Moreover, this paper introduces the novel concept of allocating a distinct confidence level to each risk, while considering the dependence structure between them.

The paper is structured as follows.
We briefly summarize univariate quantiles and expectiles in Section~\ref{section_univariate_quantiles}, while the main ideas behind the multivariate framework are introduced in Section~\ref{sec:MGeoExp}.
Specifically, Section~\ref{ssec:MGeoVaR} reviews geometric quantiles, while Section~\ref{section_geometric_expectiles} defines geometric expectiles as the main contribution of our paper.
We discuss population and asymptotic properties of the newly introduced statistical functional in Section~\ref{sec:Prop}, while examples are discussed in Section~\ref{section_illustration}.
Finally Section~\ref{section_conclusion} concludes. Selected plots can be
reproduced with the latest version of the \textsf{R} package \texttt{qrmtools}; see the vignette \texttt{geometric\_risk\_measures}.
\section{Univariate Quantiles and Expectiles}\label{section_univariate_quantiles}
It is a standard approach in statistics to express population characteristics in terms of minimizing the expected loss of a random variable under a given loss function.
Considering the absolute value $\abs{\cdot}$, the median solves $\median X = \argmin_{c \in \R} \E[\abs{X-c}]$, while the mean is obtained when considering the square loss $\E[X] = \argmin_{c \in \R} \E[(X-c)^2]$.
In case of the absolute value loss function it is readily observed that, using an asymmetric generalization of $\abs{\cdot}$, quantiles other than the median can be obtained.
For $\alpha \in (0,1)$ we define the \emph{check loss} as
\begin{align}
  \label{eq_check_function_loss}
	\rho_{\alpha} \colon \R \to [0,\infty), \quad t \mapsto \abs{\alpha - \ind{(-\infty,0]}(t)}\abs{t},
\end{align}
where we see that the case $\alpha = 0.5$ is directly related to the usual absolute value.
Similar to the median this leads to $F^{-1}(\alpha) = \argmin_{c \in \R} \E[\rho_{\alpha}(X - c)]$.
In \cite{KoenkerBassett1978} this observation is the starting point to introduce the quantile regression framework.
As an alternative \cite{NeweyPowell1987}, introduced an asymmetric version of the square loss along the same lines.
To this end we set
\begin{align}\label{eq_univariate_expectile_loss}
	\lambda_{\alpha} \colon \R \to [0,\infty), \quad t \mapsto \abs{\alpha - \ind{(-\infty,0]}(t)}\abs{t}^2,
\end{align}
 where again $\alpha \in (0,1)$.
The minimizers $\Expectile(\alpha) = \argmin_{c \in \R} \E[\lambda_{\alpha}(X - c)]$ are called expectiles, analogously to quantiles minimizing the check loss.
Again the case $\alpha = 0.5$ reduces to the well known motivating example, i.e., $\Expectile(0.5) = \E[X]$.
The generalized loss functions are asymmetric versions of their symmetric $\alpha = 0.5$ counterparts.
Compared to $\rho_{\alpha}$ the loss function $\lambda_{\alpha}$, however, is continuously differentiable, leading to favorable analytic properties in a minimization context.

\section{Multivariate Geometric Risk Measures}\label{sec:MGeoExp}
In order to generalize univariate expectiles to the multivariate setting we first revisit in Section~\ref{ssec:MGeoVaR} the framework introduced by \cite{Chaudhuri1996}.
This allows for a suitable generalization of the loss function in \eqref{eq_check_function_loss}, leading to multivariate geometric quantiles.
In Section~\ref{section_geometric_expectiles} we then apply the underlying idea of \cite{Chaudhuri1996} to give a multivariate generalization of \eqref{eq_univariate_expectile_loss} and to introduce multivariate geometric expectiles as the main contribution of this paper.
\subsection{Multivariate Geometric Value-at-Risk}\label{ssec:MGeoVaR}
\cite{Chaudhuri1996} provides a definition of multivariate quantiles by generalizing the approach outlined in Section~\ref{section_univariate_quantiles}.
The resulting geometric quantiles are obtained by minimizing the expected loss based on a multivariate loss function generalizing $\rho_\alpha$ given in \eqref{eq_check_function_loss}.

For $\bfx,\bfy \in \R^{\dim}$ we denote by $\xNorm{\bfx}{2} = \sqrt{\bfx^{\tr}\bfx}$ and $\scalarProduct{\bfx}{\bfy} = \bfx^{\tr}\bfy$ the Euclidean norm and inner product respectively, and by $B^d = \{\bfx \in {\R^{\dim}} : \xNorm{\bfx}{2} < 1 \} \subset {\R^{\dim}}$ the open unit ball in $\R^{\dim}$, where we neglect the superscript in unambiguous situations.
For a fixed index $\bfu \in B$ \cite{Chaudhuri1996} defines the loss function $\Phi_{\bfu}$ as
\begin{align}
\label{eq_multivariate_loss}
	\Phi_{\bfu} \colon \R^{\dim} \to [0,\infty), \quad \bft \mapsto \Phi_{\bfu}(\bft) = \frac{1}{2}(\xNorm{\bft}{2} + \scalarProduct{\bfu}{\bft}).
\end{align}
While it is immediately clear that $\Phi_{\bfu}(\bfzero) = 0$ for all $\bfu \in B$, we also have that $\Phi_{\bfu}(\bft) \geq 0$ for all $(\bfu,\bft) \in B \times {\R^{\dim}}$ using the Cauchy-Schwarz inequality.
Convexity of $\Phi_{\bfu}$ follows directly from properties of the norm and inner product.

Based on $\Phi_{\bfu}$ the \emph{(multivariate) geometric quantile}, or \emph{geometric $\VaR$}, at level $\bfalpha \in B$ for a random vector $\bfX$ is then defined as
\begin{align}\label{eq_multivariate_var}
  \VaR_{\bfalpha}(\bfX) = \argmin_{\bfc\in\R^{\dim}} \E[\Phi_{\bfalpha}(\bfX - \bfc)].
\end{align}
As shown in \cite{Chaudhuri1996}, the right hand side of \eqref{eq_multivariate_var} is always finite and the minimization is thus well posed.
Furthermore, the resulting geometric $\VaR$ is the unique minimizer of \eqref{eq_multivariate_var}.

In \eqref{eq_multivariate_var} the vector $\bfalpha \in B$ takes the role of the confidence level.
However, due to the multivariate context, $\VaR$ is now indexed by a $\dim$-dimensional vector instead of a scalar.
This adds additional flexibility compared to other approaches such as \cite{CousinDiBernardino2013}, \cite{Bentahar2006} and \cite{CossetteMailhotMarceauMesfioui2015}, where only one scalar confidence level can be set for the multivariate risk $\bfX$.
It is also important to notice that the geometric quantile $\VaR_{\bfalpha}(\bfX)$ itself is represented by a vector in $\R^{\dim}$.
This makes the resulting risk measure easier to use for risk analysis than approaches such as \cite{CousinDiBernardino2014}, \cite{CousinDiBernardino2013} and \cite{MailhotMoralesOmidi} where the resulting multivariate quantiles are subsets in $\R^{\dim}$.

When comparing traditional confidence levels in $(0,1)$ to the univariate case of our setting, care has to be taken to adjust the indices.
Both settings are equivalent by simply re-indexing according to $f \colon [0,1] \to [-1,1], \quad x \mapsto f(x) = 2x-1$.
An index of $99\%$ in the traditional setting is therefore comparable to an index of $98\%$ using the convention adopted in this paper.

The orientation of the contour lines of the objective function is influenced by the direction of the index $\bfu$, while the magnitude of the index changes the shape of the contour lines.
For smaller values of $\xNorm{\bfu}{2}$ the contour lines are more norm like, i.e., more circular, and in the limit $\xNorm{\bfu}{2} = 0$, i.e., if and only if $\bfu = (0,0)$ we are indeed left with the circular contour lines of the norm.
\subsection{Multivariate Geometric Expectiles}\label{section_geometric_expectiles}
Analogously to the approach of \cite{Chaudhuri1996} we introduce our multivariate representation of expectiles via a multivariate generalization of $\lambda_{\alpha}$.
For this purpose it is more convenient to rewrite the original definition of $\lambda_{\alpha}$ given in \eqref{eq_univariate_expectile_loss} as
\begin{align*}
\lambda_{\alpha}(t) = \frac{1}{2}\abs{t}(\abs{t} + (2\alpha - 1)t).
\end{align*}
It can easily be verified that both definitions coincide for all $t \in \R$.
Similarly to \eqref{eq_multivariate_loss} this motivates our definition of the loss function $\Lambda_{\bfu}$ as
\begin{align}\label{eq_multivariate_expectile_loss}
	\Lambda_{\bfu} \colon \R^{\dim} \to [0,\infty), \quad \bft \mapsto \Lambda_{\bfu}(\bft) = \frac{1}{2}\xNorm{\bft}{2}(\xNorm{\bft}{2} + \scalarProduct{\bfu}{\bft}),
\end{align}
where $\bfu \in B$ is a fixed element of the open unit ball.
Given that $\Phi_{\bfu}(\bft) \geq 0$ for all $(\bfu,\bft) \in B \times \R^{\dim}$ it is clear that we also have $\Lambda_{\bfu}(\bft) \geq 0$ for all $(\bfu,\bft) \in B \times\R^{\dim}$.
As for $\Phi_{\bfu}$ we have $\Lambda_{\bfu}(\bfzero) = 0$ for all $\bfu \in B$.

For a given confidence level $\bfalpha \in B$ we now define the \emph{geometric expectile} of a random vector $\bfX$ as the minimizer of the expected loss based on $\Lambda_{\bfalpha}$, i.e.
\begin{align}\label{eq_multivariate_expectile}
  \Expectile_{\bfalpha}(\bfX) = \argmin_{\bfc \in \R^{\dim}} \E[\Lambda_{\bfalpha}(\bfX-\bfc)].
\end{align}
As in the case of the geometric $\VaR$, the definition of geometric expectiles is based on an index $\bfalpha \in B$ allowing to specify a direction and magnitude of the confidence level.
Furthermore, geometric expectiles are vectors in $\R^{\dim}$.
This makes them easier to interpret than multivariate risk measures that are given as subsets of $\R^{\dim}$.
For example, for $\bfalpha = \bfzero$ it is easy to see that $\Expectile_{\bfzero}(\bfX) = \left(\E[X_1],\ldots,\E[X_{\dim}]\right)$.
The mean vector is therefore, analogously to the univariate case, a special case of the geometric expectiles defined in \eqref{eq_multivariate_expectile}.
In Section~\ref{sec:Prop} we discuss the existence of a minimizer $\Expectile_{\bfalpha}$ and its uniqueness  together with further properties of $\Expectile_{\bfalpha}$.

Figure~\ref{fig_contour_lines_expectiles} displays the contour lines for a two dimensional example of $\Lambda_{\bfu}(\bft)$ for three indices $\bfu_1 = 0.9/\sqrt{2}(1,1)$, $\bfu_2 = 0.9/\sqrt{2}(-1,1)$ and $\bfu_3 = 0.5/\sqrt{2}(-1,1)$.
The figure shows how the direction of the index, visualized by the arrow, changes the orientation of the contour lines (compare the left and middle plots).
Also, the magnitude of the index $\xNorm{\bfu}{2}$ influences the shape of the contour lines (compare the middle and right plots), where smaller values of $\xNorm{\bfu}{2}$ lead to more norm like contours as already discussed in the case of quantiles.

The examples in Section~\ref{section_illustration} provide numerical illustrations of the resulting expectiles for a number of bivariate distributions, see Figures~\ref{fig_numerical_expectiles} and \ref{fig_numerical_expectiles_compound_model}, as well as an analytic solution to \eqref{eq_multivariate_expectile} in the special case of a bivariate uniform distribution.
\begin{figure}
  \centering
  \includegraphics[width=0.32\textwidth]{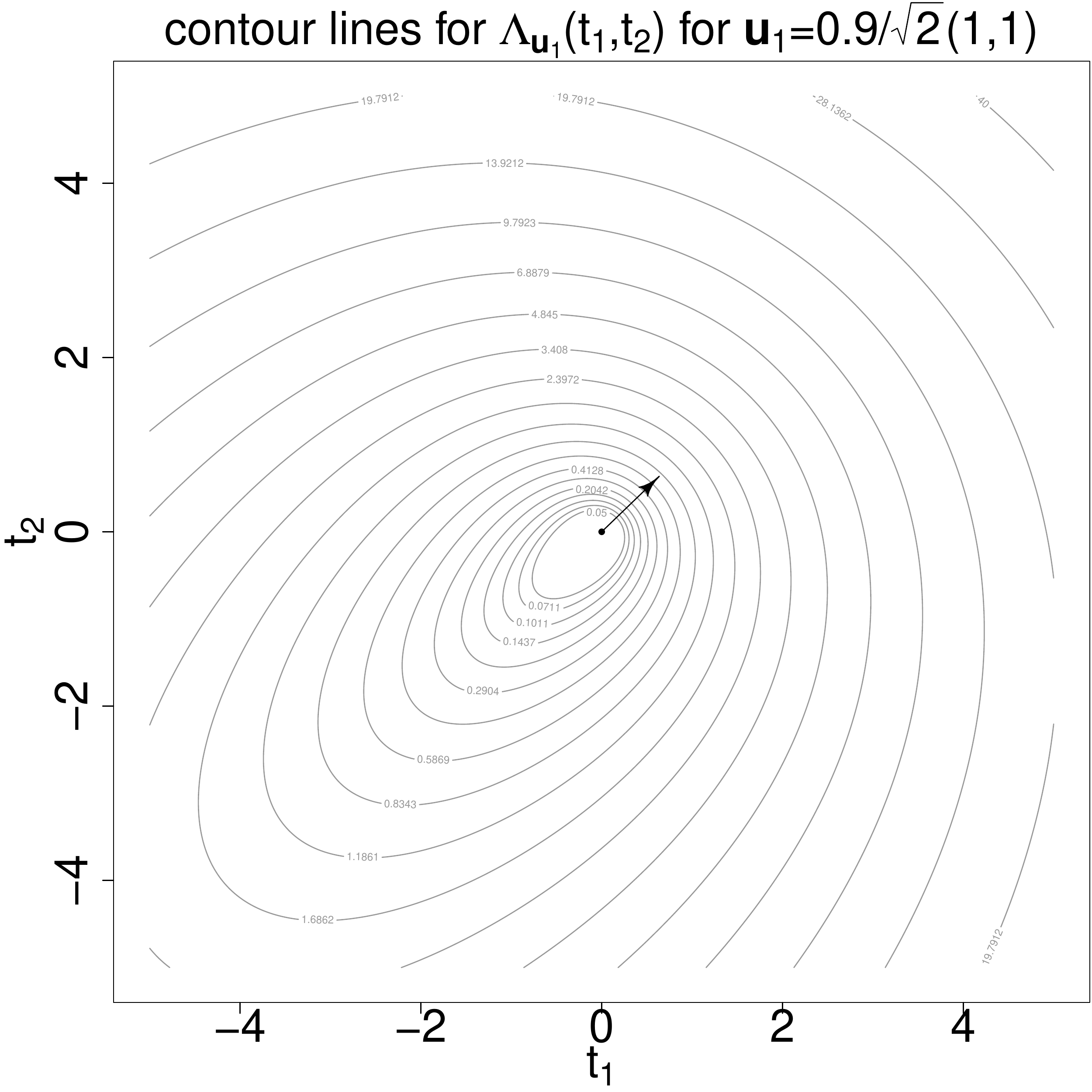}
  \includegraphics[width=0.32\textwidth]{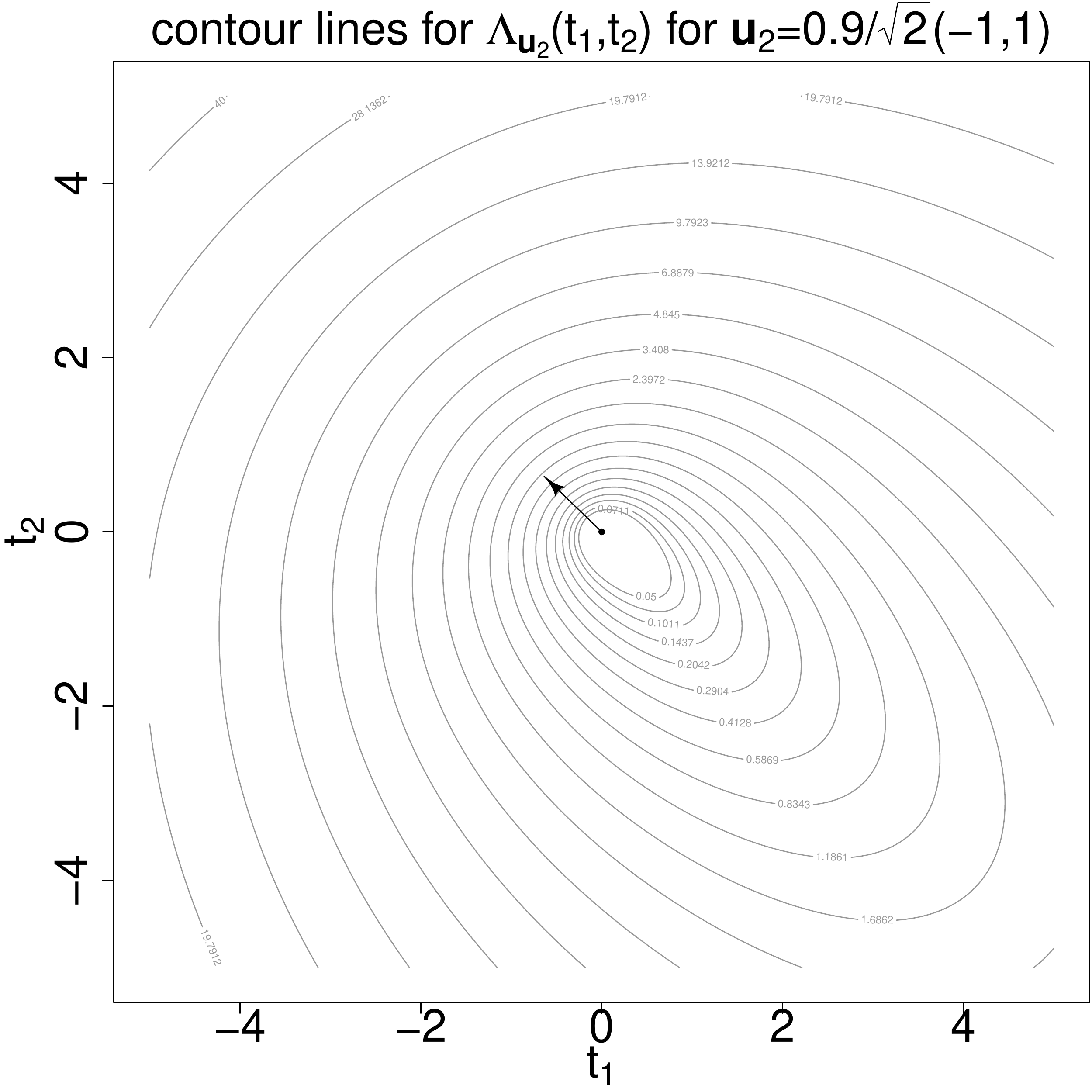}
  \includegraphics[width=0.32\textwidth]{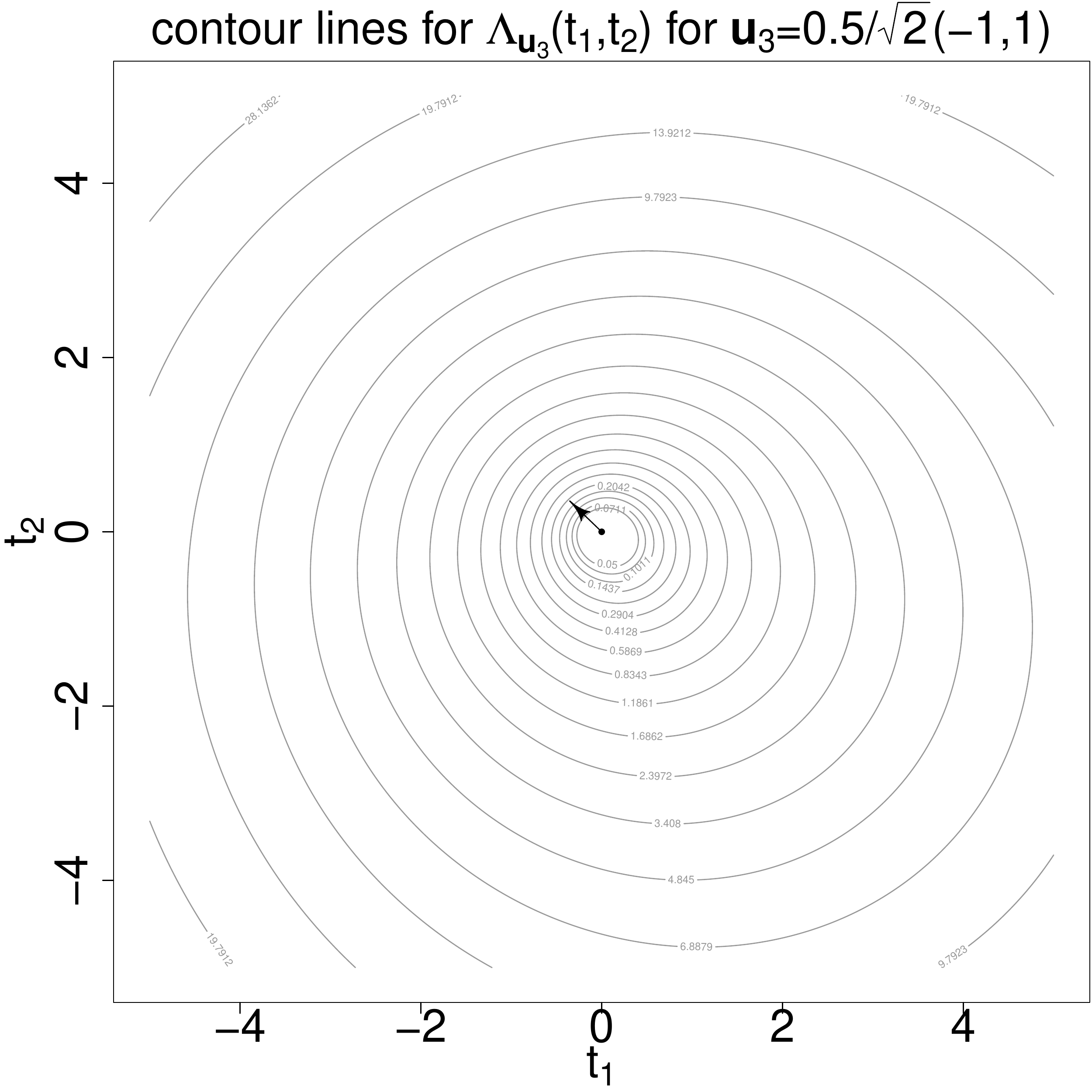}
	\caption{Contour lines for $\Lambda_{\bfu_i}(t_1,t_2)$, $i \in \{1,2,3\}$ for indices $\bfu_1=0.9/\sqrt{2}(1,1)$, $\bfu_2=0.9/\sqrt{2}(-1,1)$ and $\bfu_3=0.5/\sqrt{2}(-1,1)$. The global minimum is marked with a black dot at $(0,0)$, while the arrow visualizes direction and magnitude of the index $\bfu_i$, $i \in \{1,2,3\}$.}
    \label{fig_contour_lines_expectiles}
\end{figure}
\section{Properties of Geometric Expectiles}\label{sec:Prop}
In this section we discuss properties of geometric expectiles $\Expectile_{\bfalpha}$ defined in \eqref{eq_multivariate_expectile}.
Clearly, properties of the associated loss function $\Lambda_{\bfu}$ play a major role in this discussion which is why we discuss them first in Section~\ref{section_properties_Lambda}.
In Section~\ref{section_properties_Expectiles} we then derive properties of $\Expectile_{\bfalpha}$.
Finally, we discuss asymptotics in Section~\ref{section_estimation_asymptotics} when $\Expectile_{\bfalpha}$ needs to be estimated from observed data or approximated when closed-form solutions to the minimization problem cannot be obtained.
\subsection{Properties of $\Lambda_{\bfu}$}\label{section_properties_Lambda}
In the univariate setting an advantage of expectiles over quantiles is that the underlying loss function is differentiable at zero.
This is also true for geometric quantiles and expectiles when $\dim \geq 2$.
The following theorem shows that the geometric expectile loss function continues to be differentiable for $\dim \geq 2$, while it is straightforward to see that this is not the case for the geometric quantile loss function $\Phi_{\bfu}$ defined in \eqref{eq_multivariate_loss}.
\begin{Theorem}[Differentiability of $\Lambda_{\bfu}$]\label{theorem_differentiability_of_Lambda}
For $\Lambda_{\bfu}$ defined in \eqref{eq_multivariate_expectile_loss} the gradient $\nabla \Lambda_{\bfu}(\bft)$ exists for all $(\bfu,\bft) \in B\times\R^{\dim}$ with $\nabla\Lambda_{\bfu}(\bfzero) = \bfzero$.
\end{Theorem}
\begin{proof}
For $\bft \neq \bfzero$ it is clear that the partial derivatives with respect to each variable once exist and are finite.
To show the claim for $\bft = \bfzero$ we first consider the $k$-th element of the gradient given by
\begin{align*}
        \frac{\partial}{\partial t_k} \Lambda_{\bfu}(\bft) = t_k + \frac{t_k}{2\xNorm{\bft}{2}}\scalarProduct{\bfu}{\bft} + \frac{1}{2}\xNorm{\bft}{2}u_k.
\end{align*}
Now consider a sequence $(\bft_n)_{n=1}^{\infty}$ such that $\lim_{n\to\infty} \bft_n = \bfzero$, and we can represent each element $\bft_n$ via $\dim$-dimensional polar coordinates, i.e., we consider a radius $r_n$ and angles $\phi_{n,1},\ldots,\phi_{n,\dim-1}$ such that $\bft_n$ can be represented by
\begin{align*}
\bft_{n} = r_n
 \begin{pmatrix}
\cos(\phi_{n,1})\\
\sin(\phi_{n,1})\cos(\phi_{n,2})\\
\sin(\phi_{n,1})\sin(\phi_{n,2})\cos(\phi_{n,3})\\
\vdots\\
\sin(\phi_{n,1})\cdots\sin(\phi_{n,\dim-2})\cos(\phi_{n,\dim-1})\\
\sin(\phi_{n,1})\cdots\sin(\phi_{n,\dim-2})\sin(\phi_{n,\dim-1}),
 \end{pmatrix}
 = r_n \, \bfxi(\phi_{n,1},\ldots,\phi_{n,\dim-1})
\end{align*}
where $r_n \to 0$ as $n \to \infty$.
Writing $\xi_1,\ldots,\xi_{\dim}$ for the components of $\bfxi = \bfxi(\phi_{n,1},\ldots,\phi_{n,\dim-1})$ and noting that $\xNorm{\bfxi}{2} = 1$ we observe that
\begin{align*}
        \frac{\partial}{\partial t_k} \Lambda_{\bfu}(\bft) = r_n\xi_k + \frac{r_n\xi_k}{2r_n}r_n\scalarProduct{\bfu}{\bfxi} + \frac{1}{2}r_n\xNorm{\bm{\xi}}{2}u_k,
\end{align*}
which converges to zero for $n \to \infty$ for any sequence $(\bft_n)_{n=1}^{\infty}$ converging to zero.
\end{proof}
From the definition of $\Phi_{\bfu}$ in \eqref{eq_multivariate_loss} it is straightforward to see that $\Phi_{\bfu}$ is a convex function.
While this is also true for the loss function $\Lambda_{\bfu}$ tied to geometric expectiles, it is not immediately clear from \eqref{eq_multivariate_expectile_loss}.
To simplify the discussion we first recall a well-known result from convex analysis, see for example \cite{Rudin1976}.
\begin{Lemma}[Midpoint convexity]\label{theorem_midpoint_convexity}
Denote by $f \colon \R^{\dim} \to \R$ a continuous function.
Then $f$ is convex if and only if it is midpoint convex, i.e.
\begin{align*}
0.5 f(\bft_1) + 0.5 f(\bft_2) - f(0.5(\bft_1+\bft_2)) \geq 0
\end{align*}
for all $\bft_1,\bft_2 \in \R^{\dim}$.
\end{Lemma}
To further prepare the result we first present a theorem generalizing the familiar parallelogram identity.
While this is an essential component of our convexity proof in Theorem~\ref{theorem_strictly_convex}, the result is interesting in its own right.
\begin{Theorem}[Parallelogram Inequality]\label{theorem_parallelogram_inequality}
Denote by $\overline{B} = \{ \bfu \in \R^{\dim}: \xNorm{\bfu}{2} \leq 1\}$ the closed unit ball in $\R^{\dim}$.
For any fixed vectors $\bfx,\bfy \in \R^{\dim}$ it holds that
\begin{align}\label{eq_parallelogram_inequality}
        -\xNorm{\bfx - \bfy}{2}^2 \leq 2\xNorm{ \bfx }{2} \scalarProduct{ \bfu}{ \bfx } + 2\xNorm{ \bfy}{2} \scalarProduct{ \bfu}{ \bfy } - \xNorm{ \bfx + \bfy}{2} \scalarProduct{ \bfu}{ \bfx + \bfy } \leq \xNorm{ \bfx - \bfy }{2}^2
\end{align}
for all $\bfu \in \overline{B}$.
For $\bfu$ such that $\xNorm{ \bfu }{2} < 1$ equality holds in \eqref{eq_parallelogram_inequality} if and only if $\bfx = \bfy$.
\end{Theorem}
\begin{proof}
We start by considering the bounded components of \eqref{eq_parallelogram_inequality} as a function of $\bfu$ which can be rewritten as
\begin{align*}
f_{\bfx,\bfy}(\bfu) = \scalarProduct{ \bfu}{ (2\xNorm{ \bfx}{2} - \xNorm{ \bfx+\bfy }{2})\bfx + (2\xNorm{ \bfy}{2} - \xNorm{ \bfx+\bfy }{2})\bfy }.
\end{align*}
First, we consider two special cases.
For $\bfx = \bfy$ all terms on all sides in \eqref{eq_parallelogram_inequality} vanish and equality holds for all $\bfu \in \overline{B}$.
For $\bfx = \bfzero \neq \bfy$ we have $f_{\bfzero,\bfy}(\bfu) = \xNorm{\bfy}{2}\scalarProduct{\bfu}{\bfy}$. Therefore
\begin{align*}
-\xNorm{\bfy}{2}^2 \leq f_{\bfzero,\bfy}(\bfu) \leq \xNorm{\bfy}{2}^2
\end{align*}
holds as a consequence of the Cauchy-Schwarz inequality. Equality can only hold if $\xNorm{\bfu}{2} = 1$.
Then, for the general case, we consider $\bfy \neq \bfx \neq \bfzero$ and define
\begin{align*}
L(\bfx,\bfy) = \xNorm{ (2\xNorm{ \bfx}{2} - \xNorm{ \bfx+\bfy}{2} )\bfx + (2\xNorm{ \bfy}{2} - \xNorm{ \bfx+\bfy }{2})\bfy }{2}.
\end{align*}
With the Cauchy-Schwarz inequality and $\xNorm{ \bfu }{2} \leq 1$ we have that
\begin{align*}
-L(\bfx,\bfy) \leq f_{\bfx,\bfy}(\bfu) \leq L(\bfx,\bfy),
\end{align*}
where the equality only holds if $\xNorm{\bfu}{2} = 1$.
Our claim is now equivalent to $L(\bfx,\bfy) \leq \xNorm{ \bfx - \bfy }{2}^2$, or equivalently to
\begin{align}\label{eq_reformulation}
 \xNorm{ \bfx - \bfy }{2}^4 - L(\bfx,\bfy)^2 \geq 0.
\end{align}
Due to the scale invariance of both terms
\begin{align*}
        L(\sigma \bfx, \sigma \bfy) &= \sigma^2 L(\bfx, \bfy),\\
        \xNorm{ \sigma \bfx - \sigma \bfy }{2}^2 &= \sigma^2 \xNorm{ \bfx - \bfy }{2}^2,
\end{align*}
for any $\sigma > 0$, we consider without loss of generality $\bfx$ and $\bfy$ such that $\xNorm{\bfx+\bfy}{2} = 1$.
Any other cases can be handled by rescaling with $\sigma = 1 / \xNorm{\bfx+\bfy}{2}$.
We continue by considering polar coordinates of $( \xNorm{\bfx}{2},\xNorm{\bfy}{2}) \in \R^2$ leading to $\xNorm{\bfx}{2} = r \cos(\theta)$ and $\xNorm{\bfy}{2} = r \sin(\theta)$, where $r > 0$ and $0 \leq \theta \leq \pi/2$ due to the strict component wise positivity of $( \xNorm{\bfx}{2},\xNorm{\bfy}{2})$.
This yields
\begin{align*}
        \xNorm{\bfx+\bfy}{2}^2 = \xNorm{\bfx}{2}^2 + \xNorm{\bfy}{2}^2 + 2\scalarProduct{\bfx}{\bfy} = r^2 + 2\scalarProduct{\bfx}{\bfy} = 1,
\end{align*}
or alternatively ($2\scalarProduct{\bfx}{\bfy} = 1-r^2$).
For $\xNorm{\bfx-\bfy}{2}^2$ we have
\begin{align*}
        \xNorm{\bfx-\bfy}{2}^2 = \xNorm{\bfx}{2}^2 + \xNorm{\bfy}{2}^2 - 2\scalarProduct{\bfx}{\bfy} = r^2 - (1-r^2) = 2r^2 -1.
\end{align*}
For the first term in \eqref{eq_reformulation} we therefore have $\xNorm{\bfx-\bfy}{2}^4 = (2r^2 - 1)^2$.
Concerning $L(\bfx,\bfy)^2$ we have $L(\bfx,\bfy)^2 = \alpha^2 \xNorm{\bfx}{2}^2 + \beta^2 \xNorm{\bfy}{2}^2 + \alpha\beta 2\scalarProduct{\bfx}{\bfy}$ with $\alpha = (2\xNorm{\bfx}{2} - \xNorm{\bfx+\bfy}{2})$ and $\beta = (2\xNorm{\bfy}{2} - \xNorm{\bfx+\bfy}{2})$.
In terms of $r$ and $\theta$ we then have
\begin{align*}
        \alpha^2 &= (2r\cos(\theta)-1)^2\\
        \beta^2 &= (2r\sin(\theta)-1)^2\\
        \alpha\beta &= (2r\cos(\theta)-1)(2r\sin(\theta)-1)\\
        L(\bfx,\bfy)^2 &= \alpha^2 r^2\cos(\theta)^2 + \beta^2 r^2 \sin(\theta)^2 + \alpha\beta (1-r^2).
\end{align*}
Reformulating \eqref{eq_reformulation} in terms of $r$ and $\theta$ then yields
\begin{align*}
        \xNorm{ \bfx - \bfy }{2}^4 - L(\bfx,\bfy)^2 = 2r(r\cos(\theta)+r\sin(\theta)-1)^2 (2r\sin(\theta)\cos(\theta)+\sin(\theta)+\cos(\theta))
\end{align*}
which is non-negative given the restrictions on $\theta$.
\end{proof}
Given Lemma~\ref{theorem_midpoint_convexity} and Theorem~\ref{theorem_parallelogram_inequality} we can now establish the strict convexity of $\Lambda_{\bfu}$.
\begin{Theorem}[Strict convexity of $\Lambda_{\bfu}$]\label{theorem_strictly_convex}
  For every fixed $\bfu \in B$ the function $\Lambda_{\bfu}$ defined in \eqref{eq_multivariate_loss} is strictly convex on $\R^{\dim}$.
\end{Theorem}
\begin{proof}
Due to continuity of $\Lambda_{\bfu}$ and Lemma~\ref{theorem_midpoint_convexity} we focus on midpoint convexity.
To this end, define $D(\bfx,\bfy) = 0.5 \Lambda_{\bfu}(\bfx) + 0.5 \Lambda_{\bfu}(\bfy) - \Lambda_{\bfu}(0.5(\bfx+\bfy))$, $\bfx,\bfy \in \R^{\dim}$, where we have that $\Lambda_{\bfu}(0.5(\bfx+\bfy)) = 0.25 \xNorm{\bfx+\bfy}{2}^2 + 0.25 \xNorm{\bfx+\bfy}{2}\scalarProduct{\bfu}{\bfx+\bfy} = 0.25 \Lambda_{\bfu}(\bfx+\bfy)$.
The function $\Lambda_{\bfu}$ is then convex if and only if $h \colon \R^{\dim} \times \R^{\dim} \to \R, \quad (\bfx,\bfy) \mapsto h(\bfx,\bfy) = 4D(\bfx,\bfy) = 2 \Lambda_{\bfu}(\bfx) + 2 \Lambda_{\bfu}(\bfy) - \Lambda_{\bfu}(\bfx+\bfy)$ is non-negative.
For $h$ we have that
\begin{align*}
        h(\bfx,\bfy) &= 2 \xNorm{\bfx}{2}^2 + 2 \xNorm{\bfx}{2}\scalarProduct{\bfu}{\bfx} + 2 \xNorm{\bfy}{2}^2 + 2 \xNorm{\bfy}{2}\scalarProduct{\bfu}{\bfy} - \xNorm{\bfx+\bfy}{2}^2 - \xNorm{\bfx+\bfy}{2}\scalarProduct{\bfu}{\bfx+\bfy}\\
           &= \xNorm{\bfx-\bfy}{2}^2 + 2 \xNorm{\bfx}{2}\scalarProduct{\bfu}{\bfx} + 2 \xNorm{\bfy}{2}\scalarProduct{\bfu}{\bfy} -  \xNorm{\bfx+\bfy}{2}\scalarProduct{\bfu}{\bfx+\bfy},
\end{align*}
where we used the parallelogram identity $2 \xNorm{\bfx}{2}^2 + 2 \xNorm{\bfy}{2}^2 = \xNorm{\bfx+\bfy}{2}^2 + \xNorm{\bfx-\bfy}{2}^2$ to get the second equality.
The condition $h(\bfx,\bfy) \geq 0$ is equivalent to
\begin{align*}
-\xNorm{\bfx-\bfy}{2}^2 \leq 2 \xNorm{\bfx}{2}\scalarProduct{\bfu}{\bfx} + 2 \xNorm{\bfy}{2}\scalarProduct{\bfu}{\bfy} - \xNorm{\bfx+\bfy}{2}\scalarProduct{\bfu}{\bfx+\bfy},
\end{align*}
which holds by Theorem~\ref{theorem_parallelogram_inequality}.
The fact that $\Lambda_{\bfu}$ is strictly convex follows, as the index $\bfu$ is assumed to lie in the open ball, i.e., $\xNorm{\bfu}{2} = 1$ is not permitted.
\end{proof}
We have so far established that $\Lambda_{\bfu}$ is differentiable with a stationary point at $\bft = \bfzero$.
Furthermore, the strict convexity of $\Lambda_{\bfu}$ guarantees that there exists at most one global minimum for $\Lambda_{\bfu}$.
To finally ensure the existence of such a minimizer we establish coercivity of $\Lambda_{\bfu}$.
\begin{Definition}[Coercive function on $\R^{\dim}$]
A real valued function $f \colon \R^{\dim} \to \R$ is said to be coercive if $\lim_{n \to \infty} f(\bfx_n) = \infty$ for all sequences $(\bfx_n)_{n=1}^{\infty}$ such that $\lim_{n \to \infty} \xNorm{\bfx_n}{2} = \infty$.
\end{Definition}
Coercivity plays an important role in optimization theory as it ensures the existance of at least one minimizer for a large class of real valued functions.
This fact is formalized in the following theorem.
\begin{Theorem}\label{theorem_coercive_inf}
Denote by $f \colon \R^{\dim} \to \R$ a coercive and convex function. Then there exists an element $\bfx_0 \in \R^{\dim}$ such that $f(\bfx_0) = \inf_{\bfx \in \R^{\dim}} f(\bfx)$.
\end{Theorem}
\begin{proof}
The proof follows from the more general Theorem $2.11$ and Remark $2.13$ in \cite{BarbuPrecupanu2012} applicable to lower-semicontinuous functions on reflexive Banach spaces.
The necessary continuity of $f$ is guaranteed by Proposition $2.3$ in \cite{Tuy2016} stating that a proper convex function on $\R^{\dim}$ is continuous on every interior point of its effective domain.
\end{proof}
To finally tie all parts together we establish the coercivity of $\Lambda_{\bfu}$ in the following theorem.
Given the strict convexity of $\Lambda_{\bfu}$ established in Theorem~\ref{theorem_strictly_convex}, an application of Theorem~\ref{theorem_coercive_inf} ensures the existence of a unique and global minimizer.
From our previous observations, especially Theorem~\ref{theorem_differentiability_of_Lambda}, we know that this minimum is located at $\bfzero$ for every given $\bfu \in B$.
\begin{Theorem}[Coercivity of $\Lambda_{\bfu}$]\label{theorem_coercive_lambda}
The function $\Lambda_{\bfu}$ is coercive on $\R^{\dim}$.
\end{Theorem}
\begin{proof}
Given that $\xNorm{\bfu}{2} = s < 1$ the Cauchy-Schwarz inequality implies $\scalarProduct{\bfu}{\bfx} \geq -s\xNorm{\bfx}{2}$.
Therefore $\Lambda_{\bfu}(\bfx) \geq 0.5\xNorm{\bfx}{2}^2(1-s)$ which proofs the claim.
\end{proof}
\subsection{Properties of $\Expectile_{\bfalpha}$}\label{section_properties_Expectiles}
With the properties of $\Lambda_{\bfu}$ in place we can now tackle those of $\Expectile_{\bfalpha}$.
In \eqref{eq_multivariate_expectile}, it is necessary to ensure that the objective function, i.e., the expected loss, is finite.
Similarly to the univariate case, we recover that a finite second moment condition for the marginal distributions is sufficient.
We thus introduce the condition
\begin{itemize}[leftmargin=*]
	\item[](C) For a $\dim$-dimensional random vector $\bfX$ assume $\E[X_j^2] < \infty$ for all $j \in \{ 1,\ldots,\dim \}$
\end{itemize}
for ease of reference.
This leads to the following result.
\begin{Theorem}\label{theorem_second_order_condition}
If (C) holds for a $\dim$-dimensional random vector $\bfX = (X_1,\ldots,X_{\dim})$ then $0 \leq \E[\Lambda_{\bfu}(\bfX-\bfc)] < \infty$ for every $\bfc \in \R^{\dim}$ and $\bfu \in B$.
\end{Theorem}
\begin{proof}
We use Jensen's inequality and $\xNorm{\bfu}{2} < 1$ to obtain that
\begin{align*}
        \abs{\E[\xNorm{\bfX}{2}\scalarProduct{\bfu}{\bfX}]} \leq \E[\xNorm{\bfX}{2}\abs{\scalarProduct{\bfu}{\bfX}}] \leq \E[\xNorm{\bfX}{2}^2].
\end{align*}
For a given $\bfc \in \R^d$ and $\bfu \in B$, this leads to
\begin{align*}
        \E[\Lambda_{\bfu}(\bfX-\bfc)] &= \abs{\E[\Lambda_{\bfu}(\bfX-\bfc)]} \leq \E[\xNorm{\bfX-\bfc}{2}^2] + \E[\xNorm{\bfX-\bfc}{2}\abs{\scalarProduct{\bfu}{\bfX-\bfc}}]\\
                                    &\leq 2\E[\xNorm{\bfX-\bfc}{2}^2]\\
                                    &= 2 \sum^{\dim}_{j=1} \E[(X_j-c_j)^2] < \infty,
\end{align*}
as $\Lambda_{\bfu}(\bft) \geq 0$ for all $(\bfu,\bft) \in B\times\R^{\dim}$ and $\E[X_j^2] < \infty$.
\end{proof}
Now that the finiteness of the expected loss is addressed, we turn to the existence and uniqueness of $\Expectile_{\bfalpha}$.
To do so we adapt the proof of Theorem $6.8$ in \cite{Lehmann1983} to our more general setting.
To this end let
\begin{align}\label{eq_obj_func}
  \phi(\bfc) = \E[\Lambda_{\bfalpha}(\bfX-\bfc)]
\end{align}
denote the objective function used in \eqref{eq_multivariate_expectile} and recall the convergence in probability to infinity.
\begin{Definition}[Convergence in probability to $\infty$]
A sequence of positive random variables $(Y_n)_{n=1}^{\infty}$ converges in probability to $\infty$, if, for every $K > 0$,
$
\lim_{n\to\infty} \Prob[ Y_n > K] = 1.
$
\end{Definition}
In preparation of showing coercivity of $\phi$, we first discuss the probabilistic behaviour of $\Lambda_{\bfalpha}(\bfX - \bfc)$ when $\xNorm{\bfc}{2}$ tends towards $\infty$.
\begin{Lemma}\label{lemma_convergence_to_infinity}
  For a sequence of vectors $(\bfc_n)_{n=1}^{\infty}$ such that $\xNorm{\bfc_n}{2} \to \infty$ and for a fixed random vector $\bfX$, the sequence $(\Lambda_{\bfalpha}(\bfX - \bfc_n))_{n=1}^{\infty}$ converges in probability to $\infty$.
\end{Lemma}
\begin{proof}
  From the proof of Theorem~\ref{theorem_coercive_lambda} we have that $\Lambda_{\bfalpha}(\bft) \geq 0.5(1-s)\xNorm{\bft}{2}^2$, where $s = \xNorm{\bfalpha}{2}$. Therefore, almost surely,
\begin{align}\label{equation_inequality_1}
  \Lambda_{\bfalpha}(\bfX - \bfc) \geq 0.5(1-s)\xNorm{\bfX - \bfc}{2}^2.
\end{align}
With the reverse triangle inequality, we also have that
\begin{align}\label{equation_inequality_2}
0.5(1-s)\abs{\xNorm{\bfX}{2} - \xNorm{\bfc}{2}}^2 \leq 0.5(1-s)\xNorm{\bfX - \bfc}{2}^2,
\end{align}
almost surely.
For an arbitrary fixed $K > 0$, we define the sets
\begin{align*}
A_n(K) &= \{\omega \in \Omega : 0.5(1-s)\abs{\xNorm{\bfX}{2} - \xNorm{\bfc_n}{2}}^2 > K \},\\
B_n(K) &= \{ \omega \in \Omega : 0.5(1-s)\xNorm{\bfX - \bfc_n}{2}^2 > K\},\\
  C_n(K) &= \{ \omega \in \Omega : \Lambda_{\bfu}(\bfX - \bfc_n) > K\},
\end{align*}
leading to $A_n(K) \subset B_n(K) \subset C_n(K)$ on the basis of inequalities \eqref{equation_inequality_1} and \eqref{equation_inequality_2}.
For $A_n(K)$ and every $K>0$, we have that
\begin{align*}
\Prob[ A_n(K) ] &= 1 - \Prob[0.5(1-s)\abs{\xNorm{\bfX}{2} - \xNorm{\bfc_n}{2}}^2 \leq K],\\
&= 1 - \Prob\left[\abs{\xNorm{\bfX}{2} - \xNorm{\bfc_n}{2}} \leq \sqrt{\frac{2K}{1-s}}\right],\\
&= 1 - \( F_{\xNorm{\bfX}{2}}\(\xNorm{\bfc_n}{2}+\sqrt{\frac{2K}{1-s}}\) - F_{\xNorm{\bfX}{2}}\(\xNorm{\bfc_n}{2}-\sqrt{\frac{2K}{1-s}}\) \) \to 1,
\end{align*}
for $n \to \infty$.
Since $\Prob[A_n(K)] \leq \Prob[B_n(K)] \leq \Prob[C_n(K)]$, $\lim_{n\to\infty}\Prob[C_n(K)] = 1$ for every $K > 0$.
\end{proof}
In a second step, we show that the strict convexity of $\Lambda_{\bfu}$ established in Theorem~\ref{theorem_strictly_convex} carries over to $\phi$.
\begin{Theorem}[Strict convexity and continuity of $\phi$]\label{theorem_risk_is_convex_and_cont}
If (C) holds for a random vector $\bfX$ then $\phi \colon \R^{\dim} \to [0,\infty), \quad \bfc \mapsto \phi(\bfc) = \E[\Lambda_{\bfalpha}(\bfX-\bfc)]$ is strictly convex and continuous on $\R^{\dim}$ for every fixed $\bfalpha \in B$.
\end{Theorem}
\begin{proof}
Given that the marginal second moments of $\bfX$ are finite, Theorem~\ref{theorem_second_order_condition} guarantees that $\phi$ is well-defined for every $\bfc \in \R^{\dim}$.
With $\lambda \in [0,1]$ and $\bfc_1, \bfc_2 \in \R^{\dim}$ such that $\bfc_1 \neq \bfc_2$ we have
\begin{align*}
  \phi(\lambda \bfc_1 + (1-\lambda) \bfc_2) &=  \E[\Lambda_{\bfalpha}(\bfX - (\lambda \bfc_1 + (1-\lambda) \bfc_2))],\\
&= \E[\Lambda_{\bfalpha}(\lambda \bfX +(1-\lambda)\bfX - \lambda \bfc_1 - (1-\lambda) \bfc_2)],\\
&= \E[\Lambda_{\bfalpha}(\lambda (\bfX-\bfc_1) +(1-\lambda)(\bfX-\bfc_2))],\\
&< \E[\lambda \Lambda_{\bfalpha}(\bfX-\bfc_1)] + (1-\lambda)\Lambda_{\bfalpha}(\bfX-\bfc_2)],\\
&= \lambda \phi(\bfc_1) + (1-\lambda) \phi(\bfc_2).
\end{align*}
Continuity follows from the fact that every proper convex function on $\R^{\dim}$ is continuous on every interior point of its effective domain, see, for example, Proposition 2.3 of \cite{Tuy2016}.
Having established that $\phi$ is strictly convex on all of $\R^{\dim}$ the claim follows.
\end{proof}
Combining Lemma~\ref{lemma_convergence_to_infinity} and Theorem~\ref{theorem_risk_is_convex_and_cont} now allows us to ensure the existance and uniqueness of geometric expectiles.
\begin{Theorem}[Existence and uniqueness of $\Expectile_{\bfalpha}$]\label{theorem_existance_and_uniqueness_expectiles}
If (C) holds for a $\dim$-dimensional random vector $\bfX$ then there exists a unique solution $\Expectile_{\bfalpha}(\bfX) = \argmin_{\bfc \in \R^{\dim}} \phi(\bfc)$ for every fixed $\bfalpha \in B$.
\end{Theorem}
\begin{proof}
First, we show that $\phi$ is coercive and fix a sequence $(\bfc_n)_{n=1}^{\infty}$ such that $\xNorm{\bfc_n}{2} \to \infty$.
To show that $\phi(\bfc_n)$ diverges, we fix an arbitrary $K > 0$ and define
\begin{align*}
C_n(K) = \{\omega \in \Omega : \Lambda_{\bfalpha}(\bfX - \bfc_n) > K \}.
\end{align*}
  Given that $\Lambda_{\bfalpha}$ is positive, it follows that
\begin{align*}
  \phi(\bfc_n) &= \E[\Lambda_{\bfalpha}(\bfX-\bfc_n)]
  = \int_{C_n(K)} \Lambda_{\bfalpha}(\bfX-\bfc_n) d\Prob + \int_{\Omega \setminus C_n(K)} \Lambda_{\bfalpha}(\bfX-\bfc_n) d\Prob\\
  &\geq \int_{C_n(K)} \Lambda_{\bfalpha}(\bfX-\bfc_n) d\Prob > K \Prob[C_n(K)].
\end{align*}
Lemma~\ref{lemma_convergence_to_infinity} yields $ K \Prob[C_n(K)] \to K$ as $n \to \infty$, which shows that $\phi$ diverges, i.e., that $\phi$ is coercive.
Note that from Theorem~\ref{theorem_risk_is_convex_and_cont}, $\phi$ is also continuous and strictly convex. Then apply Theorem~\ref{theorem_coercive_inf}.
\end{proof}
Having established the basic properties of geometric expectiles, we now discuss their behaviour under data transformations.
As in \cite{Chaudhuri1996} for $\VaR_{\bfalpha}(\bfX)$, it is straightforward to show how geometric expectiles behave for translation, rotation and rescaling of the underlying random vector $\bfX$.
Adding a deterministic amount to an uncertain position simply shifts the resulting risk measure, in line with translation invariance of coherent risk measures.
\begin{Proposition}[Translation invariance]\label{proposition_translation}
If (C) holds for a $\dim$-dimensional random vector $\bfX$ then $\Expectile_{\bfalpha}(\bfX + \bfa) = \Expectile_{\bfalpha}(\bfX) + \bfa$ for all $\bfa \in \R^{\dim}$.
\end{Proposition}
\begin{proof}
By definition, we have $\Expectile_{\bfalpha}(\bfX)= \argmin_{\bfc \in \R^{\dim}} \E[\Lambda_{\bfalpha}(\bfX - \bfc)]$.
Therefore, $\E[\Lambda_{\bfalpha}(\bfX - (\bfc - \bfa))]$ will be minimized by $\Expectile_{\bfalpha}(\bfX) + \bfa$.
\end{proof}
Reasonable behaviour under scaling transformations ensures that a change in the underlying measurement units (for example, going from cents to dollars) is appropriately reflected in the behaviour of the risk measure.
For geometric expectiles this is the case as shown next.
This is in resemblance with positive homogeneity of coherent risk measures.
\begin{Proposition}[Positive homogeneity]
If (C) holds for a $\dim$-dimensional random vector $\bfX$ then $\Expectile_{\bfalpha}(\sigma \bfX) = \sigma \Expectile_{\bfalpha}(\bfX)$ for every positive scalar $\sigma > 0$.
\end{Proposition}
\begin{proof}
\begin{align*}
  \Lambda_{\bfalpha}(\sigma \bfX - \bfc) &= \tfrac{1}{2}\xNorm{\sigma \bfX - \bfc}{2}\left(\xNorm{\sigma \bfX - \bfc}{2} + \scalarProduct{\bfalpha}{\sigma \bfX-\bfc}\right)\\
                 &= \tfrac{1}{2}\sigma \xNorm{\bfX - \sigma^{-1} \bfc}{2}\left(\sigma \xNorm{\bfX - \sigma^{-1} \bfc}{2} + \sigma \scalarProduct{\bfalpha}{\bfX- \sigma^{-1} \bfc}\right)\\
                 &= \sigma^2 \Lambda_{\bfalpha}(\bfX - \sigma^{-1} \bfc).
\end{align*}
  Given that $\Expectile_{\bfalpha}(\bfX)$ minimizes $\sigma^2 \E[\Lambda_{\bfalpha}(\bfX - \bfc)]$, as the positive factor $\sigma^2$ only changes the value of the objective function but not the location of the optimum, we have that $\sigma \Expectile_{\bfalpha}(\bfX)$ minimizes $\E[\Lambda_{\bfalpha}(\sigma \bfX - \bfc)]$, i.e., $\Expectile_{\bfalpha}(\sigma \bfX) = \sigma \Expectile_{\bfalpha}(\bfX)$.
\end{proof}
It is reasonable to expect that a permutation of the components of $\bfX$ should likewise result in a permutation of the entries of the resulting risk measure.
Geometric expectiles are not only well behaved under permutations, but under general orthogonal rotations.
\begin{Proposition}[Rotation with orthogonal matrix]\label{proposition_rotation}
If (C) holds for a $\dim$-dimensional random vector $\bfX$ then $\Expectile_{A\bfalpha}(A\bfX) = A\Expectile_{\bfalpha}(\bfX)$ for every orthogonal matrix $A \in \R^{\dim \times \dim}$.
\end{Proposition}
\begin{proof}
	By orthogonality of $A$, for every $\bfx,\bfy \in \R^{\dim}$, we have that $\scalarProduct{A\bfx}{A\bfy} = \scalarProduct{\bfx}{\bfy}$ and $\xNorm{A\bfx}{2} = \xNorm{\bfx}{2}$.
Denoting by $A^{\tr}$ the transpose of $A$ we therefore get
\begin{align*}
  \Lambda_{A\bfalpha}(A\bfX - \bfc) &= \tfrac{1}{2}\xNorm{A\bfX - \bfc}{2}\left(\xNorm{A\bfX - \bfc}{2} + \scalarProduct{A\bfalpha}{A\bfX-\bfc}\right)\\
                           &= \tfrac{1}{2}\xNorm{\bfX - A^{\tr} \bfc}{2}\left(\xNorm{\bfX - A^{\tr} \bfc}{2} + \scalarProduct{\bfalpha}{\bfX-A^{\tr}\bfc}\right)\\
                           &= \Lambda_{\bfalpha}(\bfX - A^{\tr} \bfc).
\end{align*}
  Given that $\Expectile_{\bfalpha}(\bfX)$ minimizes $\E[\Lambda_{\bfalpha}(\bfX - \bfc)]$, the minimizer of $\E[\Lambda_{\bfalpha}(\bfX - A^{\tr} \bfc)]$ is given by $A\Expectile_{\bfalpha}(\bfX)$.
\end{proof}
Taken together, Propositions~\ref{proposition_translation}--\ref{proposition_rotation} guarantee that $\Expectile_{\bfalpha}(\bfX)$ is well behaved for the most relevant data transformations.
In this context it is natural to ask if there exists a suitable ordering $\prec$ for random vectors such that $\bfX \prec \bfY$ implies $\Expectile_{\bfalpha}(\bfX) \sqsubset \Expectile_{\bfalpha}(\bfY)$ for a possibly different ordering $\sqsubset$.
While this point is of great interest it proofed too difficult to establish a suitable result and we thus leave it as an open question for further research.

In the following Corollary~\ref{proposition_negative_sign} and Proposition~\ref{proposition_symmetry} we generalize well known symmetry properties of univariate expectiles to the multivariate setting.
We start by establishing a connection between the geometric expectiles of $\bfX$ and $-\bfX$ as a corolloary of Proposition~\ref{proposition_rotation}.
\begin{Corollary}[Vector sign symmetry]\label{proposition_negative_sign}
If (C) holds for a $\dim$-dimensional random vector $\bfX$ then $\Expectile_{\bfalpha}(-\bfX) = -\Expectile_{-\bfalpha}(\bfX)$ for all $\bfalpha \in B$.
\end{Corollary}
\begin{proof}
Apply Proposition~\ref{proposition_rotation} with $A = -I$, where $I$ is the appropriate identity matrix, and $-\bfalpha$ and re-arrange terms.
\end{proof}
For radially symmetric distributions, see for example \cite{McNeilFreyEmbrechts2015} Chapter 7, the resulting expectiles also obey a symmetry relation when changing the sign of the underlying index $\bfalpha$.
\begin{Proposition}[Index sign symmetry]\label{proposition_symmetry}
If (C) holds for a $\dim$-dimensional radially symmetric random vector $\bfX$ with mean vector $\bfmu$, then $\bfmu = \tfrac{1}{2}\left(\Expectile_{\bfalpha}(\bfX) + \Expectile_{-\bfalpha}(\bfX)\right)$ for all $\bfalpha \in B$.
\end{Proposition}
\begin{proof}
For $\bfalpha \in B$ we have
\begin{align*}
\E[\Lambda_{-\bfalpha}(\bfX - \bfc)] &= \tfrac{1}{2} \E[\xNorm{\bfX - \bfc}{2}^2] + \tfrac{1}{2}\E[\xNorm{\bfX - \bfc}{2}\scalarProduct{-\bfalpha}{\bfX-\bfc}]\\
	&= \tfrac{1}{2}\E[\xNorm{-(\bfc - \bfX)}{2}^2] + \tfrac{1}{2}\E[\xNorm{-(\bfc-\bfX)}{2}\scalarProduct{\bfalpha}{\bfc-\bfX}]\\
	&= \E[\Lambda_{\bfalpha}(\bfc-\bfX)]\\
	&= \E[\Lambda_{\bfalpha}(-(\bfX - \bfc))]
\end{align*}
	for all $\bfc \in \R^{\dim}$.
	By translation invariance (Proposition \ref{proposition_translation}) and radial symmetry $\bfX - \bfmu \eqdist -(\bfX-\bfmu)$ we therefore have
	\begin{align*}
		\E[\Lambda_{-\bfalpha}(\bfX - (2\bfmu - \Expectile_{\bfalpha}(\bfX)))] &=
		\E[\Lambda_{\bfalpha}(-(\bfX - \bfmu) - \Expectile_{\bfalpha}(\bfX)+\bfmu)]\\
		&=\E[\Lambda_{\bfalpha}((\bfX - \bfmu) - (\Expectile_{\bfalpha}(\bfX)-\bfmu))]\\
&=\E[\Lambda_{\bfalpha}((\bfX - \bfmu) - \Expectile_{\bfalpha}(\bfX-\bfmu))],
	\end{align*}
	where the right hand side is minimized implying that $\Expectile_{-\bfalpha}(\bfX) = 2\bfmu -\Expectile_{\bfalpha}(\bfX)$.
\end{proof}

In the univariate setting, expectiles are an attractive choice among possible risk measures due to their elicitability.
As discussed in \cite{Gneiting2011}, elicitability is a property of statistical functionals when considering point forecasts.
Denoting by $\mathcal{F}$ the class of probability distributions on $\R^{\dim}$ with finite second marginal moments, we denote by $T$ a statistical functional, i.e.,
\begin{align*}
T \colon \mathcal{F} \to \R^{\dim}, \quad F \mapsto T(F).
\end{align*}
Statistical functionals can, in general, be set valued maps, as for example in the case of quantiles.
However, we will here concentrate on the case where they take values in the Euclidean space, and we thus adjust the definition of elicitability given in \cite{Gneiting2011} to this case.
\begin{Definition}[Elicitability]
A statistical functional $T$ is called \emph{elicitable} relative to the class $\mathcal{F}$, if
\begin{itemize}
  \item[1)] there exists a scoring function
$
S \colon \R^{\dim} \times \R^{\dim} \to [0,\infty), \quad (\bfx,\bfy) \mapsto S(\bfx,\bfy)
$
such that there is a representation
\begin{align*}
  T(F) = \argmin_{\bfc \in \R^{\dim}}\E[S(\bfc,\bfX)],
\end{align*}
for every $F \in \mathcal{F}$ where $\bfX \sim F$, and
\item[2)]$\E[S(T(F),\bfX)]=\E[S(\bfc,\bfX)]$ implies $\bfc = T(F)$.
\end{itemize}
\end{Definition}
A functional $T$ is therefore elicitable, if it can be represented as the unique minimizer of a Bayes rule for a suitable scoring function.
For geometric expectiles, we can define, for $\bfalpha \in B$, an associated functional $T_{\bfalpha}$ as
\begin{align*}
T_{\bfalpha} \colon \mathcal{F} \to \R^{\dim}, \quad F \mapsto \Expectile_{\bfalpha}(\bfX) \mbox{ for } \bfX \sim F,
\end{align*}
where Theorem~\ref{theorem_existance_and_uniqueness_expectiles} guarantees that $T_{\bfalpha}(F)$ is not set-valued.
It is then clear from the defintion of $\Expectile_{\bfalpha}$ that the scoring function
\begin{align*}
  S_{\bfalpha} \colon \R^{\dim} \times \R^{\dim}, \quad (\bfx,\bfy) \mapsto \Lambda_{\bfalpha}(\bfx-\bfy)
\end{align*}
makes $T_{\bfalpha}$ elicitable relative to the class $\mathcal{F}$.
Again here, Theorem~\ref{theorem_existance_and_uniqueness_expectiles} plays a crucial role under the assumption of a joint distribution with margins with finite second moments.

In the univariate case elicitability allows to assess and compare the forecasting performance of different competing models, see \cite{NoldeZiegel2017} for a discussion.
In a practical setting this allows one to select a best model based on expectile point-forecasting performance and to implement meaningful expectile-based backtesting procedures against real data.
Elicitability of geometric expectiles now possibly opens the door to implement model selection and backtesting procedures for the underlying joint distribution as opposed to the the marginal distributions only.
From a theoretical perspective, geometric expectiles also add to a further understanding of multivariate elicitability by provding a scoring function that is not a linear combination of univariate scoring functions; see \cite{FisslerZiegel2016} for an in-depth discussion.

The scoring function $S_{\bfalpha}$ tied to geometric expectiles is furthermore positively homogeneous of order two as shown in Proposition~\ref{proposition_homogenious} below.
\cite{Efron1991} highlights the necessity of positive homogeneity, or scale invariance, in an estimation context.
Scale invariance and estimation of scale is also central to the theory of robust statistics; see, for example, \cite{HuberRonchetti2009}.
\cite{Patton2011} furthermore argues for homogeneity in the context of forecast rankings, as the rankings obtained from a homogenious scoring function are invariant to a re-scaling of the underlying data.
See also \cite{Gneiting2011} and \cite{NoldeZiegel2017} for a discussion in the context of univariate expectlies.

By establishing the positive homogeneity of $S_{\bfalpha}$ we prepare likewise applications in the multivariate case.
\begin{Proposition}[Positive homogeneity of $S_{\bfalpha}$ of order $2$]\label{proposition_homogenious}
For $c > 0$ and $(\bfx,\bfy) \in \R^{\dim} \times \R^{\dim}$, $S_{\bfalpha}(c\bfx,c\bfy) = c^2 S_{\bfalpha}(\bfx,\bfy)$.
\end{Proposition}
\begin{proof}
Using basic properties of norms and inner products we have that
\begin{align*}
  S_{\bfalpha}(c\bfx,c\bfy) &= \Lambda_{\bfalpha}(c(\bfx-\bfy)) = \frac{1}{2}\xNorm{c(\bfx-\bfy)}{2}(\xNorm{c(\bfx-\bfy)}{2} + \scalarProduct{\bfu}{c(\bfx-\bfy)})\\
  &= c^2\frac{1}{2}\xNorm{\bfx-\bfy}{2}(\xNorm{\bfx-\bfy}{2} + \scalarProduct{\bfu}{\bfx-\bfy}) = c^2 \Lambda_{\bfalpha}(\bfx-\bfy) = c^2 S_{\bfalpha}(\bfx,\bfy). \qedhere
\end{align*}
\end{proof}

Univariate expectiles are attractive risk measures due to their coherence of which sub-additivity is a cornerstone.
For univariate expectiles we have for any random variables $X$ and $Y$ sub-additivity $\Expectile_{\alpha}(X+Y) \leq \Expectile_{\alpha}(X) + \Expectile_{\alpha}(Y)$ when $\alpha \geq 0.5$, while for $\alpha \leq 0.5$ we have super-additivity $\Expectile_{\alpha}(X+Y) \geq \Expectile_{\alpha}(X) + \Expectile_{\alpha}(Y)$.
It is important to recognize that $\Expectile_{0.5}(X) = \E[X]$, i.e. there is one point which separates the sub- and super-additive cases.

While the univariate notions of sub- and superadditivity are based on the ordering in $\R$, the multivariate case has no canonical ordering for $\R^{\dim}$, $\dim \geq 2$.
To circumvent this issue we utilize set inclusions that continue to be valid in higher dimensions.
Reconsidering the univariate case we can see that for any interval $I \subseteq (0,1)$ that includes $0.5$ we have $\{x \in \R : x = \Expectile_{\alpha}(X+Y), \alpha \in I\} \subseteq \{x \in \R : x = \Expectile_{\alpha}(X)+\Expectile_{\alpha}(Y), \alpha \in I\}$.
To propose a multivariate generalization based on this observation we replace the interval $I$ with a closed ball in $B$.
\begin{Definition}[Multivariate subadditivity for geometric risk measures]
Denote by $\bfX$ and $\bfY$ two $\dim$-dimensional random vectors, and by $\rho_{\bfalpha}$ a geometric risk measure based on an index $\bfalpha \in B$.
For $0 < r < 1$ define the sets
\begin{align*}
A_r(\bfX;\rho) &= \{\bfx \in \R^{\dim} : \bfx = \rho_{\bfalpha}(\bfX), \xNorm{\bfalpha}{2} \leq r\},\\
A_r(\bfX,\bfY;\rho) &= \{\bfx \in \R^{\dim} : \bfx = \rho_{\bfalpha}(\bfX) + \rho_{\bfalpha}(\bfY), \xNorm{\bfalpha}{2} \leq r\}.
\end{align*}
A multivariate geometric risk measure $\rho_{\bfalpha}$ is \emph{multivariate sub-additive}, if
\begin{align*}
	A_r(\bfX+\bfY;\rho) \subseteq A_r(\bfX,\bfY;\rho)
\end{align*}
for all $0 < r < 1$.
\end{Definition}
We use the numerical techniques discussed in Section~\ref{section_illustration} for a two dimensional illustration.
To this end we introduce a random vector $\bfZ = (Z_1,\ldots,Z_4)$, where the first marginal distribution $Z_1$ follows a Gumbel distribution, $Z_2 \sim t_4$, $Z_3$ follows a standard logistic distribution and $Z_4 \sim \N(0,1)$.
To introduce dependence between the components of $\bfZ$ we join the margins by a four dimensional Clayton copula $\Copula_{\theta}$ with parameter $\theta = 5$.
The bivariate random vectors are then given as $\bfX = (Z_1,Z_2)$ and $\bfY = (Z_3,Z_4)$.

In Figure~\ref{fig_sub_super_additivity} (left) we show $A_{0.2}(\bfX+\bfY;\VaR)$ and $A_{0.2}(\bfX,\bfY;\VaR)$ where it is clearly visible that geometric $\VaR$ is not multivariate sub-additive which is in line with the univariate case.
This behaviour can be explained when focusing on the case $\bfalpha = \bfzero$, in which case geometric $\VaR$ is the minimizer of the euclidean distance
\begin{align*}
  \VaR_{\bfzero}(\bfX) = \argmin_{\bfc\in\R^{\dim}} \E[\xNorm{\bfX - \bfc}{2}].
\end{align*}
There is no reason that the resulting optimum is additive, i.e., $\VaR_{\bfzero}(\bfX+\bfY) = \VaR_{\bfzero}(\bfX) + \VaR_{\bfzero}(\bfY)$.
Given that the sets $A_{r}(\bfX+\bfY;\VaR)$ and $A_{r}(\bfX,\bfY;\VaR)$ reduce to $\VaR_{\bfzero}(\bfX+\bfY)$ and $\VaR_{\bfzero}(\bfX)+\VaR_{\bfzero}(\bfY)$ when $\bfalpha \to \bfzero$, the sets necessarily intersect for some $r$ whenever $\VaR_{\bfzero}(\bfX+\bfY) \neq \VaR_{\bfzero}(\bfX) + \VaR_{\bfzero}(\bfY)$.
This behaviour is shown on the left in Figure~\ref{fig_sub_super_additivity}.

In Figure~\ref{fig_sub_super_additivity} (right) we show $A_{0.2}(\bfX+\bfY;\Expectile)$ and $A_{0.2}(\bfX,\bfY;\Expectile)$.
In this case we observe $A_{0.2}(\bfX+\bfY;\Expectile) \subseteq A_{0.2}(\bfX,\bfY;\Expectile)$.
Contrary to geometric $\VaR$ we have $\Expectile_{\bfzero}(\bfX) = \E[\bfX]$ in the case of geometric expectiles and therefore the additivity $\Expectile_{\bfzero}(\bfX+\bfY) = \Expectile_{\bfzero}(\bfX) + \Expectile_{\bfzero}(\bfY)$.
Constructing a counter example to multivariate sub-additivity along the same lines as for $\VaR$ is therefore ruled out.
Although numerical checks for a number of different joint models and $r$-levels suggest that geometric expectiles are multivariate subadditive a formal proof is not available at this point.

\begin{figure}
  \centering
  \includegraphics[width=0.49\textwidth]{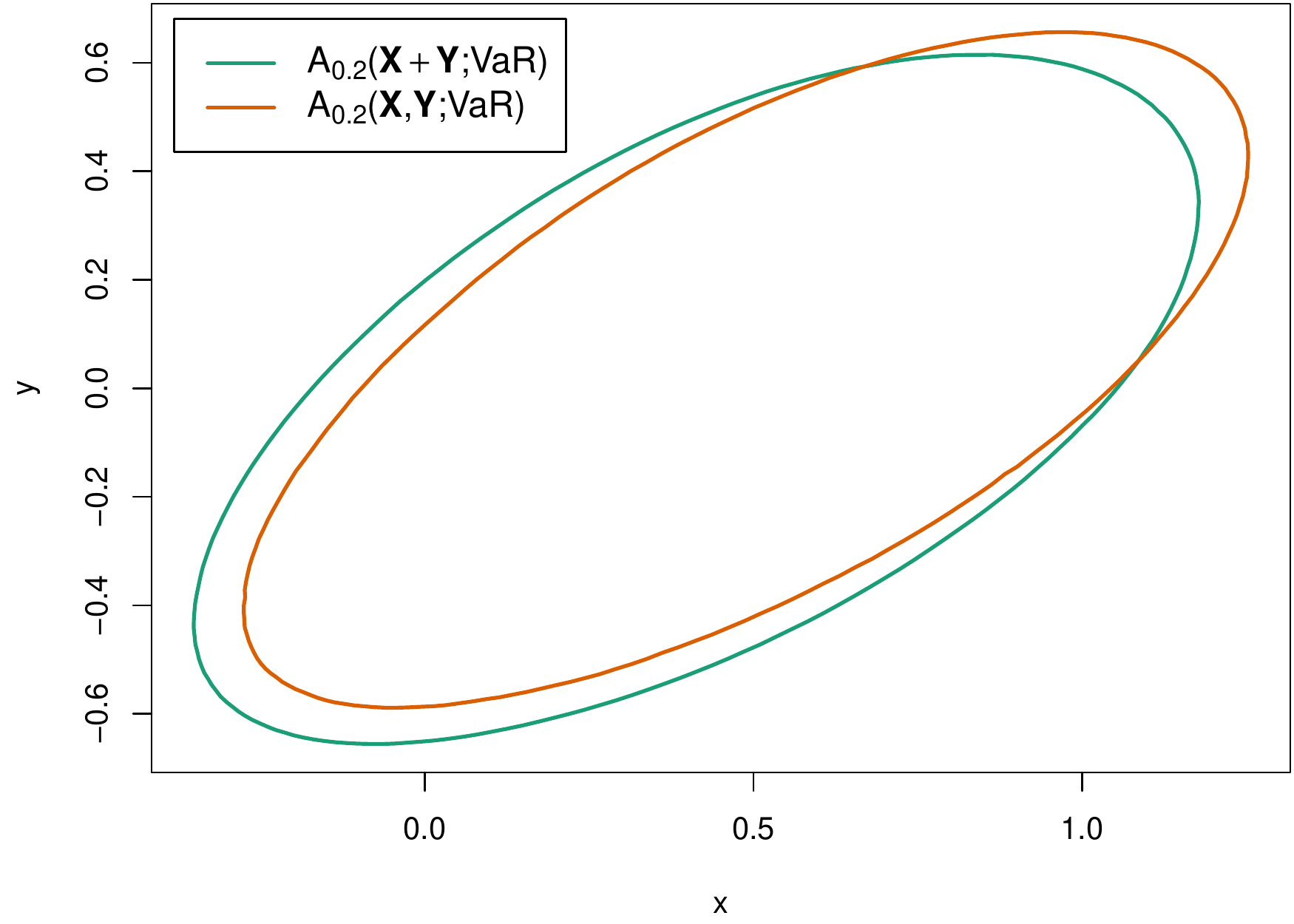}
  \includegraphics[width=0.49\textwidth]{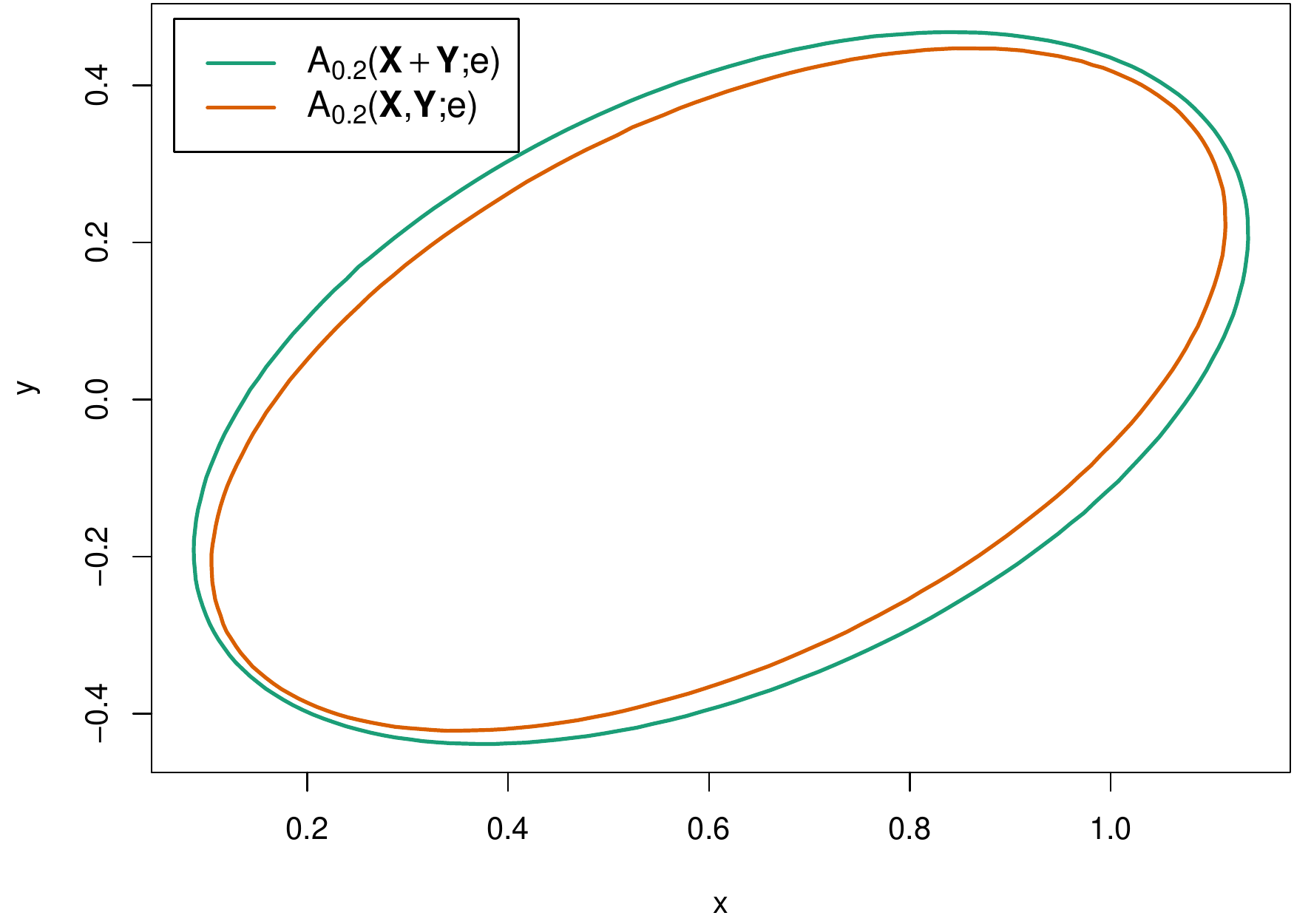}
\caption{Boundary of the sets $A_{0.2}(\bfX+\bfY;\VaR)$ (left, boundary in green), $A_{0.2}(\bfX,\bfY;\VaR)$ (left, boundary in orange), and $A_{0.2}(\bfX+\bfY;\Expectile)$ (right, boundary in green), $A_{0.2}(\bfX,\bfY;\Expectile)$ (right, boundary in orange). $\bfX = (Z_1,Z_2)$ and $\bfY = (Z_3,Z_4)$, where $Z_1$ follows a Gumbel distribution, $Z_2 \sim t_4$, $Z_3$ follows a standard logistic distribution and $Z_4 \sim \N(0,1)$. All margins are joined by a four dimensional Clayton copula $\Copula_{\theta}$ with parameter $\theta = 5$. Computations are based on $20000$ independent replications.}
    \label{fig_sub_super_additivity}
\end{figure}

\subsection{Asymptotics and Estimation}\label{section_estimation_asymptotics}
In Section~\ref{sec:Prop} we have established that geometric expectiles defined in \eqref{eq_multivariate_expectile} are a well-defined functional for random vectors with finite marginal second moments.
In terms of practical applications, this raises two questions.
First, the computation of closed-form solutions of $\Expectile_{\bfalpha}(\bfX)$ might not be possible for a given random vector $\bfX$ and numerical approximation needs to be invoked instead.
Second, in practical applications it is necessary to establish that a sample version of  $\Expectile_{\bfalpha}(\bfX)$ is a consistent estimator of $\Expectile_{\bfalpha}(\bfX)$.
While the implicit definition of $\Expectile_{\bfalpha}(\bfX)$ might seem challenging at first, our functional falls into the well-established framework of M-estimators; see \cite{HuberRonchetti2009} for an introduction.

To discuss consistency, we denote by $(\bfX_i)_{i=1}^{\infty}$ a sequence of independent and identically distributed (iid) random vectors with the same distribution as $\bfX$.
While a generalization to ergodic and (weakly) stationary random vectors is straight forward, we focus on the iid case for ease of presentation.
For a finite sample we replace the expectation in \eqref{eq_multivariate_expectile} by the sample average.
This provides a finite sample version, or Monte Carlo estimator, of $\phi$ defined in \eqref{eq_obj_func} by
\begin{align*}
\phi_n(\bfc) = \frac{1}{n} \sum_{i=1}^n \Lambda_{\bfalpha}(\bfX_i - \bfc),
\end{align*}
where we immediately get $\phi_n(\bfc) \as \phi(\bfc)$ (and thus $\phi_n(\bfc) \inprob \phi(\bfc)$) from the strong law of large numbers.
To also guarantee the convergence of the minimizers we invoke Proposition 7.4 of \cite{Hayashi2000}.
\begin{Corollary}[Consistency]\label{corolloary_consistency}
If (C) holds for a $\dim$-dimensional random vector $\bfX$ then
\begin{align*}
\argmin_{\bfc \in \R^d} \phi_n(\bfc) \inprob \Expectile_{\bfalpha}(\bfX).
\end{align*}
\end{Corollary}
\begin{proof}
We apply Proposition 7.4 of \cite{Hayashi2000} which guarantees the consistency of M-estimators.
By Theorem~\ref{theorem_existance_and_uniqueness_expectiles} $\phi(\bfc)$ is uniquely minimized on $\R^{\dim}$ at $\Expectile_{\bfalpha}(\bfX)$, and $\phi(\bfc)$ exists and is finite for all $\bfc \in \R^{\dim}$.
Furthermore, $\Lambda_{\bfalpha}$ is convex.
While the existence of a minimizer in Proposition $7.4$ in \cite{Hayashi2000} is only asymptotic, it is clear that a minimizer exists for every $n \in  \NN$ in our case.
\end{proof}
Corollary~\ref{corolloary_consistency} also suggests a simple approach to compute $\Expectile_{\bfalpha}(\bfX)$ when a close form solution cannot be established.
If a sampling method for $\bfX$ is available then replacing the expectation by an empirical mean yields a valid approximation.
Denoting by $(\bfx_n)_{i=1}^n$ a realization of a sequence of random vectors $(\bfX_i)_{i=1}^n$, either obtained by simulation or from real data, it also important to notice that $\phi_n$ is strictly convex.
\begin{Corollary}[Strict convexity of $\phi_n$]\label{corolloary_sample_convexity}
Denote by $(\bfx_i)_{i=1}^n$ a sequence of vectors in $\R^{\dim}$.
The function $\phi_n \colon \R^d \to \R, \quad \bfc \mapsto \phi_n(\bfc) = \frac{1}{n} \sum_{i=1}^n \Lambda_{\bfalpha}(\bfx_i - \bfc)$ is strictly convex.
\end{Corollary}
\begin{proof}
Given that $\phi_n$ is a convex combination of strictly convex functions the proof follows from basic properties of convex functions.
\end{proof}
The importance of Corollary~\ref{corolloary_sample_convexity} is that the minimization
\begin{align*}
\argmin_{\bfc\in\R^{\dim}} \phi_n(\bfc)
\end{align*}
is well behaved also in the finite sample case and there exists a unique minimizer that is consistent for the functional according to Corollary~\ref{corolloary_consistency}.
The minimizer, i.e., the finite sample version of $\Expectile_{\bfalpha}(\bfX)$, can then be obtained by numerical minimization techniques.
\section{Illustration}\label{section_illustration}
In this section, we discuss a special case for which it is possible to obtain a closed-form expression for multivariate geometric expectiles.
Moreover, we provide numerical illustrations for a number of different random vectors in order to highlight the impact of changing margins and dependence structures.
\subsection{Analytic Solution for the Uniform Distribution}
We consider the case of a bivariate uniform distribution and denote by $\bfU = (U_1,U_2)$ a random vector with density $\frac{1}{(b_1-a_1)(b_2-a_2)}$ where $b_j > a_j$ and $b_j,a_j \in \R$ for $j=1,2$.
We first compute the expectation of the squared norm in terms of $\bfc = (c_1,c_2)$ as
\begin{align*}
	g(c_1,c_2) &=
        \E[\xNorm{\bfU-\bfc}{2}^2] =
        \int_{a_1}^{b_1} \int_{a_2}^{b_2} \frac{(u_1-c_1)^2 + (u_2-c_2)^2}{(b_1-a_1)(b_2-a_2)} du_2 du_1\\
        &= \frac{1}{(b_1-a_1)(b_2-a_2)} \int_{a_1-c_1}^{b_1-c_1} \int_{a_2-c_2}^{b_2-c_2} u_1^2 + u_2^2 du_2 du_1\\
        &= \frac{(b_2-a_2)((b_1-c_1)^3-(a_1-c_1)^3) + (b_1-a_1)((b_2-c_2)^3-(a_2-c_2)^3)}{3(b_1-a_1)(b_2-a_2)}.
\end{align*}
Defining the real valued functions $h_1$ and $h_2$ as
\begin{align*}
h_1(x,y) &= \frac{1}{2} \(y \sqrt{x^2+y^2} + x^2 \log\(y+\sqrt{x^2+y^2}\)\)\quad\mbox{and}\quad\\
h_2(x,y) &= \frac{1}{96} \(-3 x^4 + 20 x^2 y \sqrt{x^2+y^2} + y^3 \(3 y + 8 \sqrt{x^2+y^2}\) + 12 x^4 \log\(y + \sqrt{x^2+y^2}\)\),
\end{align*}
we further have that
\begin{align*}
        \frac{d}{dx} h_2(x,y) = x h_1(x,y) \quad \mbox{and that} \quad \frac{d}{dy} h_1(x,y) = \sqrt{x^2 + y^2}.
\end{align*}
Therefore,
\begin{align*}
        \int_{a_1}^{b_1} x \int_{a_2}^{b_2} \sqrt{x^2 + y^2} dy dx &= \int_{a_1}^{b_1} x (h_1(x,b_2) - h_1(x,a_2)) dx\\
                                                                   &= \int_{a_1}^{b_1} x h_1(x,b_2) dx - \int_{a_1}^{b_1} x h_1(x,a_2) dx\\
                                                                   &= h_2(b_1,b_2) - h_2(a_1,b_2) - h_2(b_1,a_2) + h_2(a_1,a_2).
\end{align*}
This finally leads to
\begin{align*}
	&g_1(c_1,c_2) =
        \E[\xNorm{\bfU-\bfc}{2}(U_1-c_1)] = \int_{a_1}^{b_1} \int_{a_2}^{b_2} \frac{(u_1-c_1) \sqrt{(u_1-c_1)^2 + (u_2-c_2)^2}}{(b_1-a_1)(b_2-a_2)} du_2 du_1,\\
        &= \frac{1}{(b_1-a_1)(b_2-a_2)} \int_{a_1-c_1}^{b_1-c_1} u_1 \int_{a_2-c_2}^{b_2-c_2} \sqrt{u_1^2 + u_2^2} du_2 du_1,\\
        &= \frac{h_2(b_1-c_1,b_2-c_2) - h_2(a_1-c_1,b_2-c_2) - h_2(b_1-c_1,a_2-c_2) + h_2(a_1-c_1,a_2-c_2)}{(b_1-a_1)(b_2-a_2)},
\end{align*}
where we define $g_2$ analogously as $g_2(c_1,c_2) = \E[\xNorm{\bfU-\bfc}{2}(U_2-c_2)]$.
Taking the preceding results together, we have for $\bfalpha = (\alpha_1,\alpha_2)$ that
\begin{align*}
  \phi(\bfc) &= \E[\Lambda_{\bfalpha}(\bfU-\bfc)]
                = \frac{1}{2}\E[\xNorm{\bfU-\bfc}{2}^2 + \alpha_1\xNorm{\bfU-\bfc}{2}(U_1-c_1) + \alpha_2\xNorm{\bfU-\bfc}{2}(U_2-c_2)]\\
		&= \frac{g(c_1,c_2)}{2} + \frac{\alpha_1 g_1(c_1,c_2)}{2}  + \frac{\alpha_2 g_2(c_1,c_2)}{2}.
\end{align*}
The geometric expectiles $\Expectile_{\bfalpha}(\bfU)$ are now found as
\begin{align*}
        \Expectile_{\bfalpha}(\bfU) = \argmin_{\bfc \in \R^2} \phi(\bfc).
\end{align*}
This example highlights more than anything that finding a closed-form solution can be challenging even in the simplest of cases.
In this sense, the numerical approximation introduced in Section~\ref{section_estimation_asymptotics} takes a more prominent role.
Full-fledged examples utilizing this method can be found in the following sections.
\subsection{Numerical Illustration}\label{section_numerical_illustration}
In this section we visualize geometric expectiles for selected bivariate random vectors.
To this end, we define four random vectors $\bfX_1,\ldots,\bfX_4$ with different margins and dependence structures; see Table~\ref{table_example_specification}.
The dependence structure is formalized in terms of copulas, see, for example, \cite{nelsen2006} or \cite{Joe2014} for textbook introductions.
As a baseline for our comparison, $\bfX_1 = (X_{11},X_{12})$ follows a bivariate normal distribution with independent standard normal margins.
Considering $\bfX_2$ we keep the independence between the components, but we change the margins.
$X_{21}$ now follows a skew normal distribution, see \cite{Azzalini1985}, with parameters $(\xi, \omega, \alpha) = (-1,1,2)$ and  $X_{22}$ follows a Student~$t$ distribution with $\nu = 4$ degrees of freedom.
In case of $\bfX_3$ we only change the dependence structure compared to $\bfX_1$, that is $X_{31}$ and $X_{32}$ still follow a standard normal distribution each but the dependence structure is now given by a Gumbel copula with parameter $\theta = 2$.
Finally $\bfX_4$ differs from $\bfX_1$ in terms of margins and dependence structure, where we employ the skew normal and Student~$t$ margins of $\bfX_2$ with the Gumbel dependence structure of $\bfX_3$.
\begin{table}
\centering
\begin{tabular}{clll}
\hline\noalign{\smallskip}
Vector & Copula & $X_{i1} \sim$ & $X_{i2} \sim$\\
\hline\hline\noalign{\smallskip}\vspace{0.4cm}
$\bm{X}_1 = (X_{11},X_{12})$ & Independence & $\mathcal{N}(0,1)$ & $\mathcal{N}(0,1)$\\
\vspace{0.4cm}
$\bm{X}_2 = (X_{21},X_{22})$ & Independence & $\mathcal{SN}(-1,1,2)$ & $t_4$\\
\vspace{0.4cm}
$\bm{X}_3 = (X_{31},X_{32})$ & Gumbel, $\theta=2$ & $\mathcal{N}(0,1)$ & $\mathcal{N}(0,1)$\\
\vspace{0.4cm}
$\bm{X}_4 = (X_{41},X_{42})$ & Gumbel, $\theta=2$ & $\mathcal{SN}(-1,1,2)$ & $t_4$\\
\hline
\end{tabular}
\caption{Specification of the random vectors $\bfX_1,\ldots,\bfX_4$.}
  \label{table_example_specification}
\end{table}

To illustrate the impact of different indices we consider two parameterizations for $\bfalpha$.
First, we choose $\bfalpha$ according to $\bfalpha_1(\varphi) = 0.98 (\cos(\varphi), \sin(\varphi))$, $\varphi \in [0,2\pi)$, which describes a circle of radius $0.98$.
The magnitude $0.98$ corresponds to a confidence level of $0.99$ in the univariate case.
Second, we choose $\bfalpha$ according to $\bfalpha_2(\varphi) = (0.98 \cos(\varphi), 0.90 \sin(\varphi))$, $\varphi \in [0,2\pi)$, which describes an ellipse in $B$, where a magnitude of $0.90$ corresponds to a confidence level of $0.95$ in the univariate case.
Both choices of indices are visualized in Figure~\ref{fig_possible_directions}.
For further reference we indicate the resulting indices $\bfalpha_j(\varphi_k)$, $j\in\{1,2\}$, for $\varphi_k = k 2\pi/8 $, $k\in\{0,\ldots,7\}$, by the respective value of $k$.
As there are no closed-form solutions available to compute $\Expectile_{\bfalpha_j(\varphi)}(\bfX_\indI)$, $\indI \in \{1,2,3,4\}$, we instead draw an iid sample of size $10,000$ from the respective distribution of $\bfX_\indI$ and utilize the numerical procedure outlined in Section~\ref{section_estimation_asymptotics}; i.e., we use Monte Carlo integration.
\begin{figure}
\centering
\includegraphics[width=0.4\textwidth]{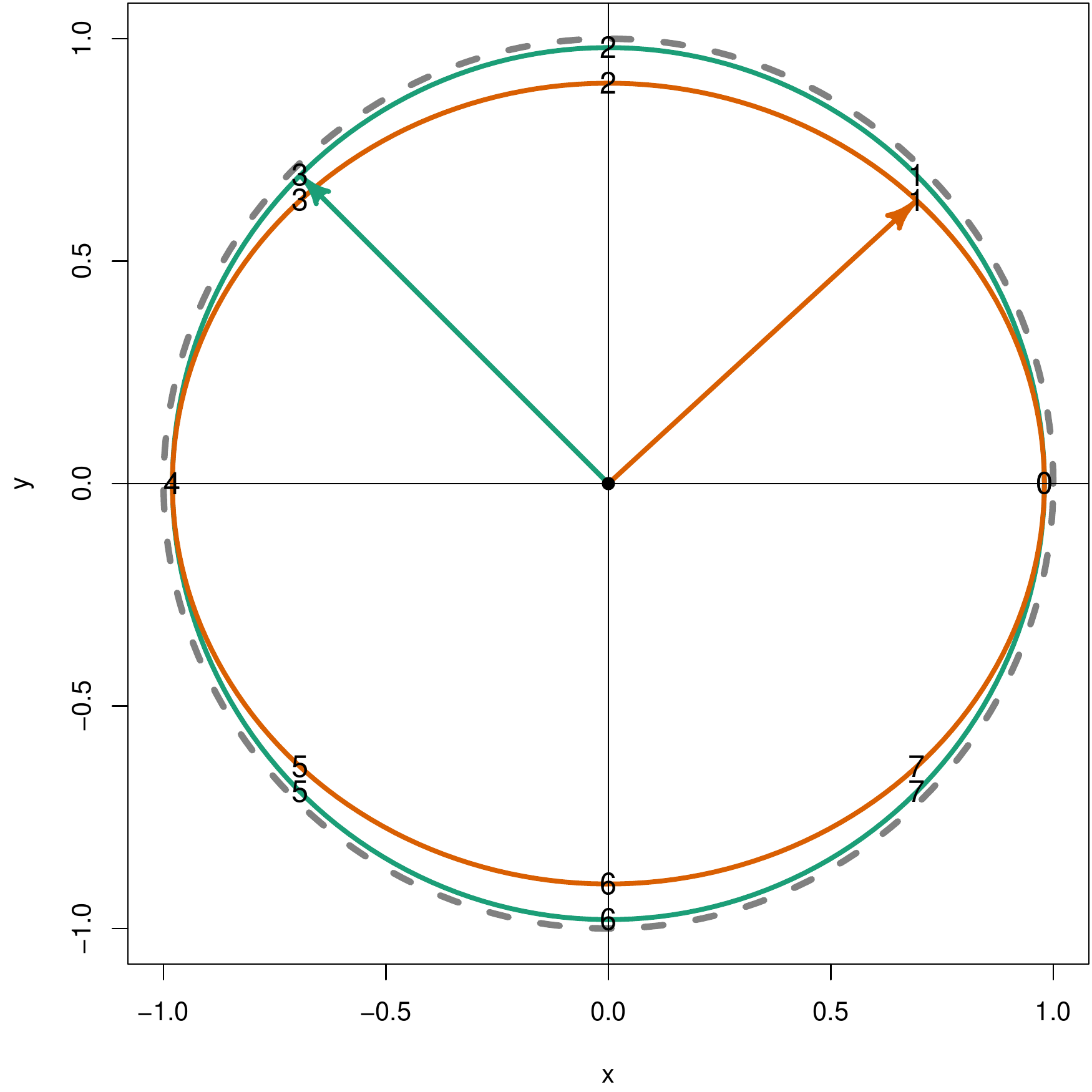}
	\caption{Bivariate indices $\bfalpha_1(\varphi) = 0.98(\cos(\varphi),\sin(\varphi))$ for $\varphi = 3\pi/4$ (green arrow) and $\bfalpha_2(\varphi) = (0.98\cos(\varphi),0.90\sin(\varphi))$ (orange arrow) for $\varphi = \pi/4$. The solid lines indicate possible indices when $\varphi$ varies in $[0,2\pi)$. Numbers indicate $\bfalpha_j(\varphi_k)$, $j \in \{1,2\}$, for $\varphi_{k} = k 2\pi / 8$ where $k \in \{0,\ldots,7\}$. The gray dashed line indicates the boundary of the open unit ball $\partial B = \{\bfx \in \R^2 : \xNorm{\bfx}{2} = 1\}$.
}
\label{fig_possible_directions}
\end{figure}

Figure~\ref{fig_numerical_expectiles} shows the resulting geometric expectiles and density contour lines for $\bfX_1$ (top left), $\bfX_2$ (top right), $\bfX_3$ (bottom left) and $\bfX_4$ (bottom right).
The gray lines indicate the density contours of the underlying bivariate distribution function.
To indicate the effects of different index choices the solid orange line represents the resulting geometric expectiles $\Expectile_{\bfalpha_1(\varphi)}(\bfX_\indI)$, $\indI \in \{1,\ldots,4\}$, for $\varphi \in [0,2\pi)$.
Likewise the solid green line indicates $\Expectile_{\bfalpha_2(\varphi)}(\bfX_\indI)$, $\indI \in \{1,\ldots,4\}$, for $\varphi \in [0,2\pi)$.
In concordance with Figure~\ref{fig_possible_directions} we mark the resulting geometric expectiles $\Expectile_{\bfalpha_j(\varphi_k)}(\bfX_i)$ for indices $\bfalpha_j(\varphi_k)$ based on $\varphi_k = k 2\pi/8 $, $k \in \{0,\ldots,7\}$, by the respective value of $k$.
\begin{figure}
\centering
\includegraphics[width=0.49\textwidth]{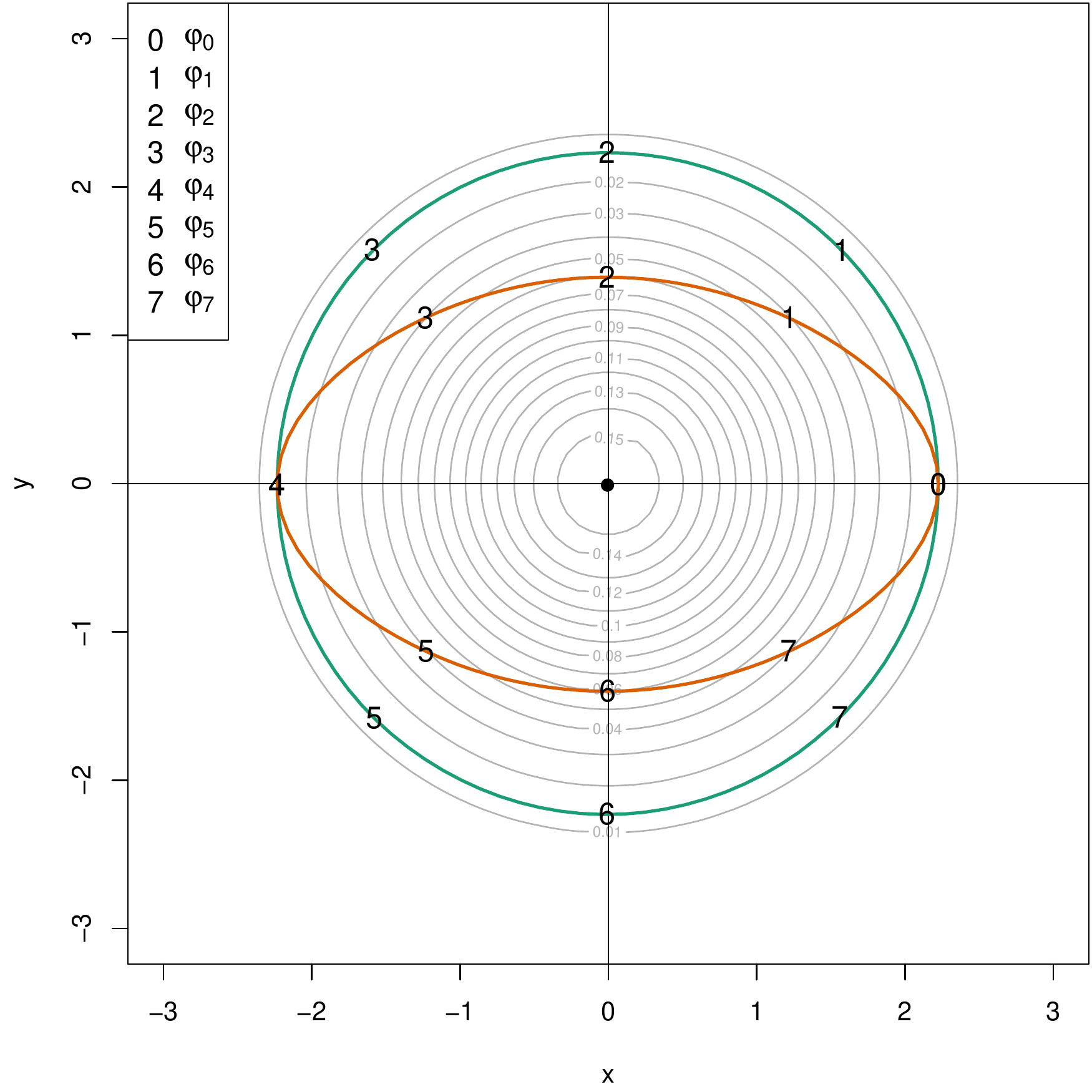}
\includegraphics[width=0.49\textwidth]{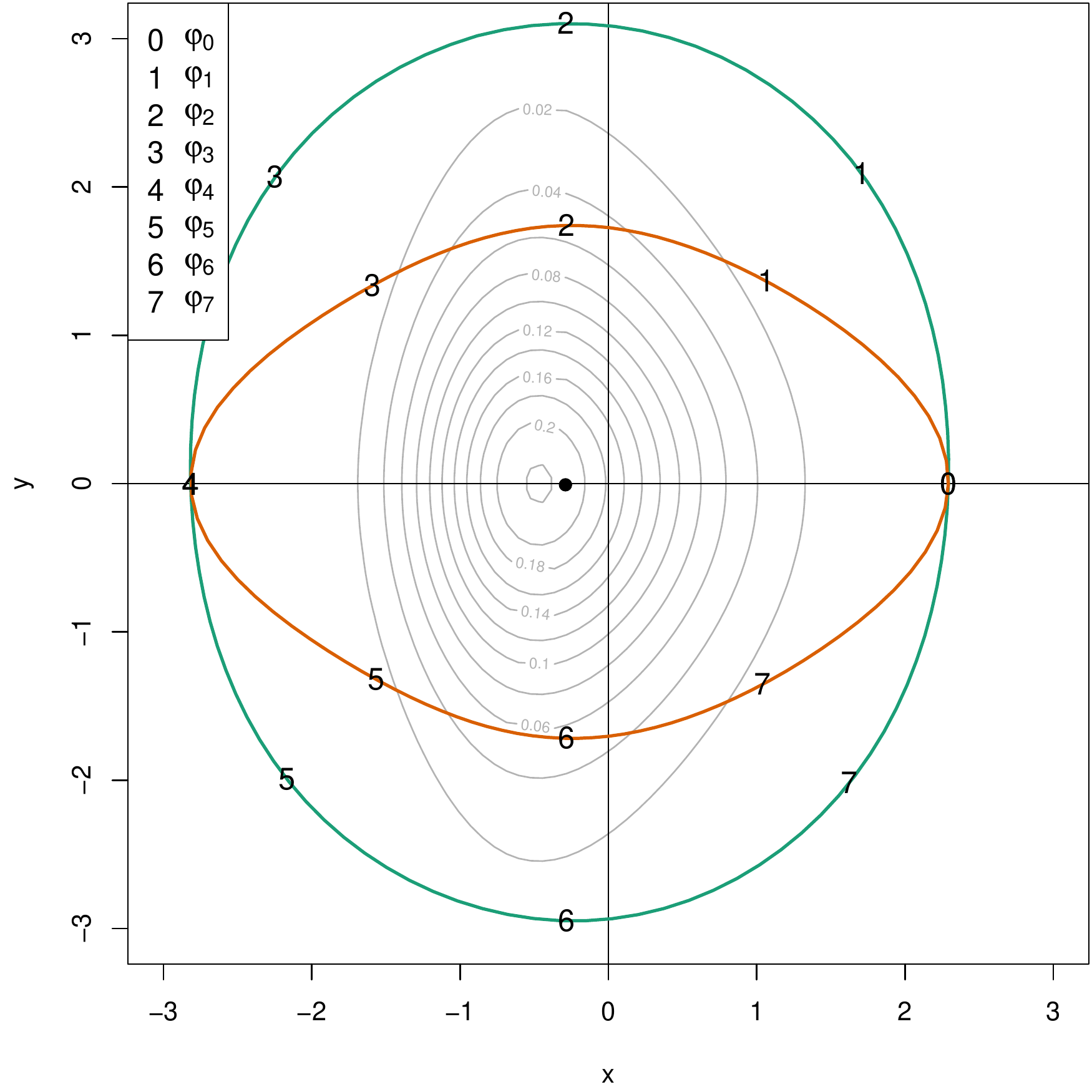}
\includegraphics[width=0.49\textwidth]{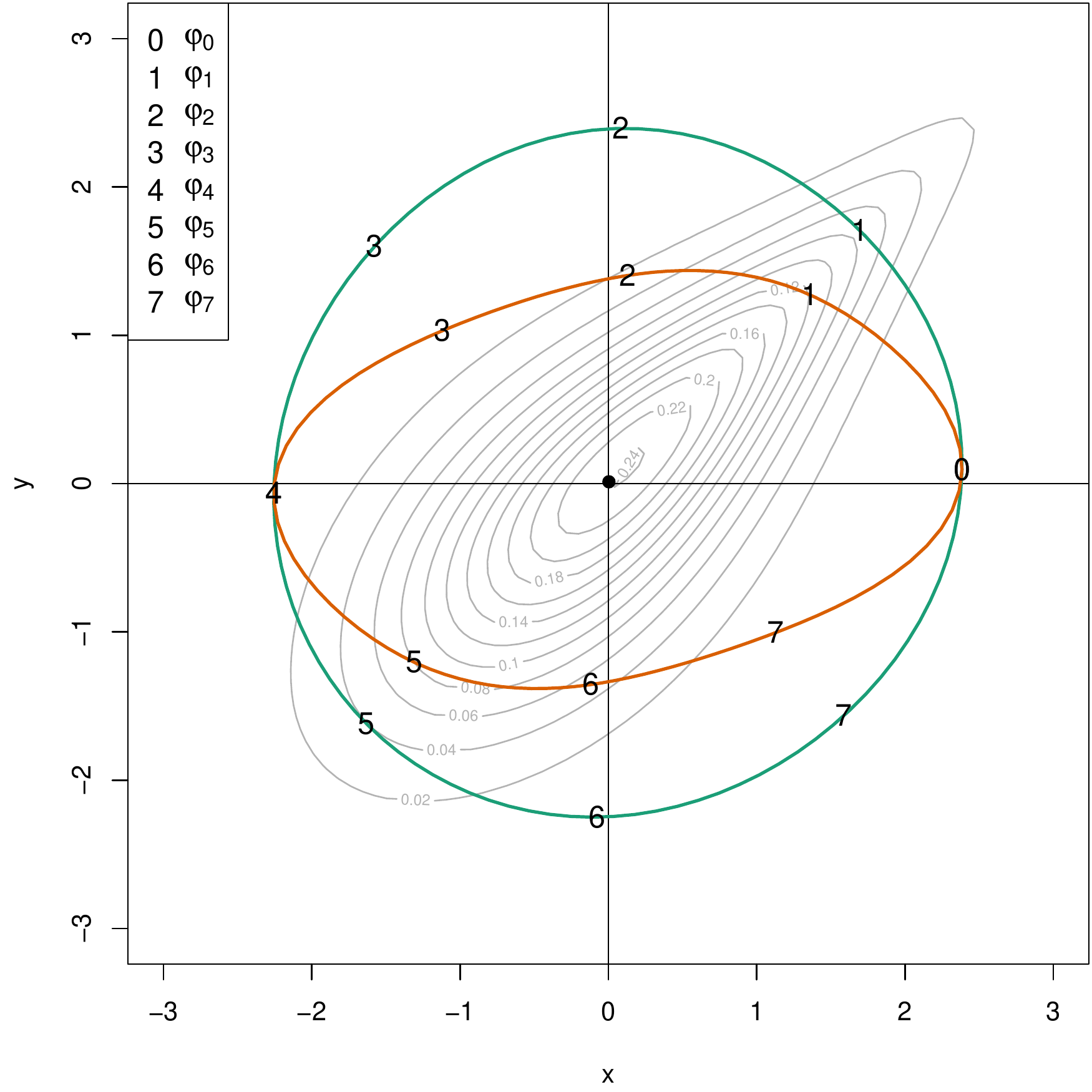}
\includegraphics[width=0.49\textwidth]{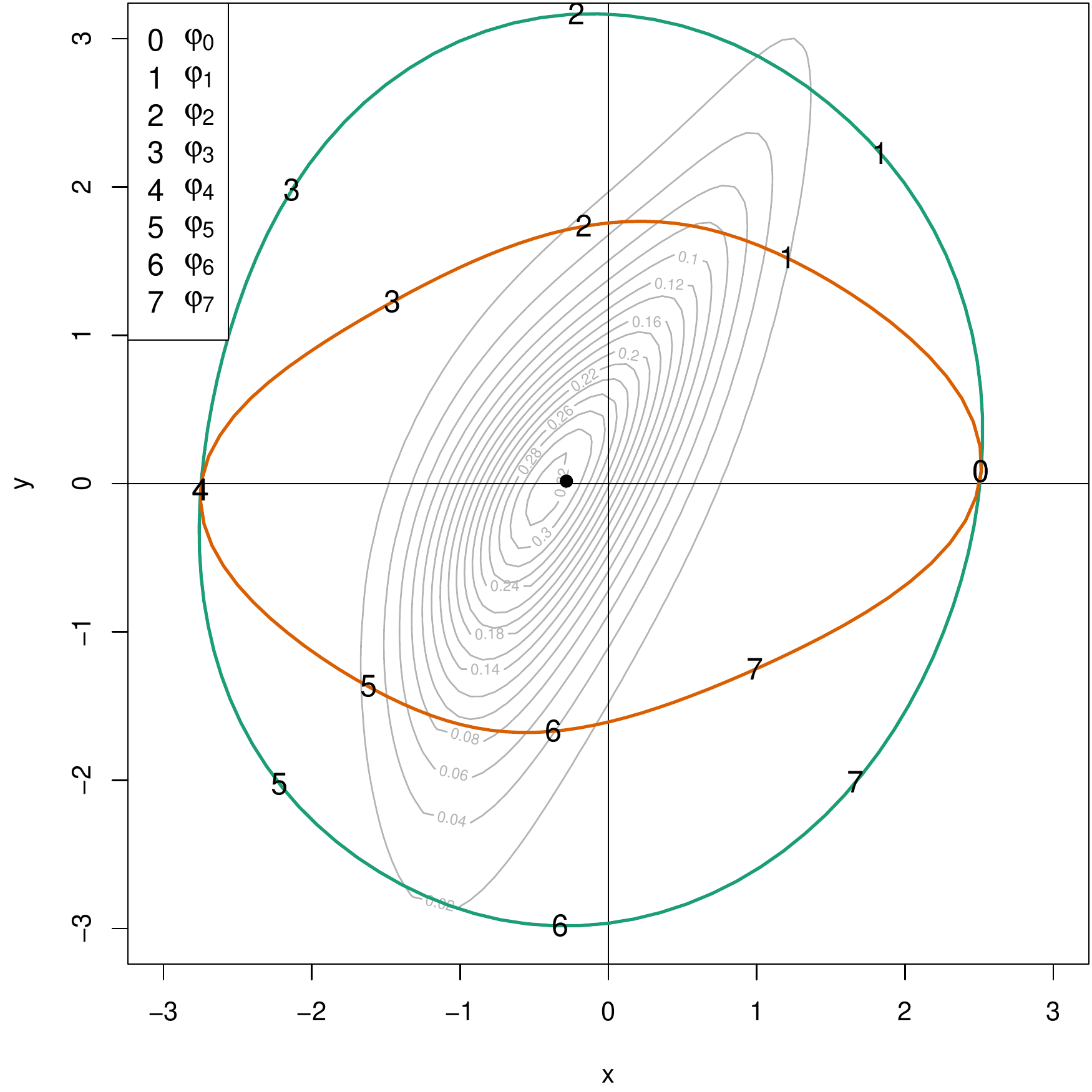}
	\caption{Geometric expectiles $\Expectile_{\bfalpha_j}(\bfX_\indI)$ for $\indI=1$ (top left), $\indI=2$ (top right), $\indI=3$ (bottom left) and $\indI=4$ (bottom right) and $j \in \{1,2\}$. The density contour lines of $\bfX_\indI$ are given in gray. The solid green lines indicate $\Expectile_{\bfalpha_1(\varphi)}(\bfX_\indI)$, $\varphi \in [0,2\pi)$, for $\bfalpha_1(\varphi) = 0.98(\cos(\varphi),\sin(\varphi))$. The solid orange lines indicate $\Expectile_{\bfalpha_2(\varphi)}(\bfX_\indI)$, $\varphi \in [0,2\pi)$, for $\bfalpha_2(\varphi) = (0.98\cos(\varphi),0.90\sin(\varphi))$. Numbers indicate $\Expectile_{\bfalpha_j(\varphi_k)}(\bfX_\indI)$ for $\varphi_{k} = k 2\pi / 8$ where $k\in\{0,\ldots,7\}$. The mean vector $(\E[X_{\indI,1}],\E[X_{\indI,2}])$ is represented by the black dot. Computations are based on $10,000$ iid realizations of $\bfX_\indI$ for each $\indI\in\{1,\ldots,4\}$.}
\label{fig_numerical_expectiles}
\end{figure}

From Figure~\ref{fig_numerical_expectiles} it becomes apparent that geometric expectiles adapt to the underlying distribution.
For the radially symmetric distribution of $\bfX_1$ (top left panel) the lines indicating $\Expectile_{\bfalpha_j(\varphi)}(\bfX_1)$ for all possible $\varphi \in [0,2\pi)$ resemble the shape of the index $\bfalpha_j(\varphi)$.
Furthermore we visually observe the symmetry established in Proposition \ref{proposition_symmetry}.
However, for skewed and heavier tailed margins (top right panel) the geometric expectiles adapt by bulging out.
This also slightly changes the orientation in that, for example, $\Expectile_{\bfalpha_1(\varphi_2)}(\bfX_2)$ is not centered on the $y$-axis anymore.
Introducing dependence between the components of $\bfX_3$ (bottom left panel) forces the geometric expectiles to deform.
The deformation, compared to the top left panel, is, however, not by bulging out as in the top right panel, but rather by compressing and rotating.
Finally when combining both effects in $\bfX_4$ (bottom right panel) we see that geometric expectiles widen and deform according to a superposition of the previously observed effects.

\subsection{Comparing Geometric Value-at-Risk and Expectiles}\label{sub_section_Comparing}
In continuation of the numerical examples in Section~\ref{section_numerical_illustration} we now discuss differences between geometric $\VaR$ and geometric expectiles, as well as their univariate counterparts.
For a fixed $\alpha_1 \in (0,1)$ we therefore consider the corresponding index $\bfalpha = (2\alpha_1-1) (1,0)$, where we make the necessary adjustment to the magnitude of the index discussed in Section~\ref{ssec:MGeoVaR}.
We then compute the univariate $\VaR_{\alpha_1}(X_{11})$ and expectile $\Expectile_{\alpha_1}(X_{11})$ at level $\alpha_1$ for the first component of $\bfX_1 = (X_{11},X_{12})$, see Table~\ref{table_example_specification}, and also the geometric $\VaR_{\bfalpha}(\bfX_1)$ and geometric expectile $\Expectile_{\bfalpha}(\bfX_1)$ based on $\bfalpha$.
Comparing the univariate risk measures to the first component of their multivariate counterparts in Figure~\ref{fig_compare_first_entry_normal}, we see that the multivariate risk measures are more conservative, i.e., higher in absolute value.
In fact, the geometric $\VaR$ provides the most conservative reserve estimates for a given level $\alpha_1$, while the univariate expectiles are the least conservative for the same level.

We further compare geometric $\VaR$ and geometric expectiles by computing the magnitude for a given direction that leads to equal values in each component of the resulting multivariate risk measure.
We therefore fix an element $\bfu \in \partial B = \{\bfx \in \R^{2} : \xNorm{\bfx}{2} = 1\}$ and $\theta \in [0,1)$ to obtain a starting index $\bfalpha = \theta \bfu$ for which we compute $\Expectile_{\bfalpha}(\bfX_1)$.
For $m \in [0,1)$ we then aim to find an optimal $m^*$ that yields $\VaR_{m^*\bfu}(\bfX_1) = \Expectile_{\bfalpha}(\bfX_1)$ in the least-square sense, that is
\begin{align*}
  m^* = \argmin_{m \in [0,1)} \sum^{2}_{i=1} \left( \VaR_{m\bfu}(\bfX_1)_i -  \Expectile_{\bfalpha}(\bfX_1)_i \right)^2.
\end{align*}
In Figure~\ref{fig_equal_magnitude} we show the resulting plot for values of $\theta$ in $[0,0.999]$ and $\bfu = (1,1)/\sqrt{2}$.
In line with Figure~\ref{fig_compare_first_entry_normal} we find that the associated magnitude $m^*$ for the geometric $\VaR$ is lower than the corresponding magnitude $\theta$ for the geometric expectiles.
That is to say that the geometric $\VaR$ is more conservative than geometric expectiles in this example.
\begin{figure}
\centering
  \includegraphics[width=\textwidth]{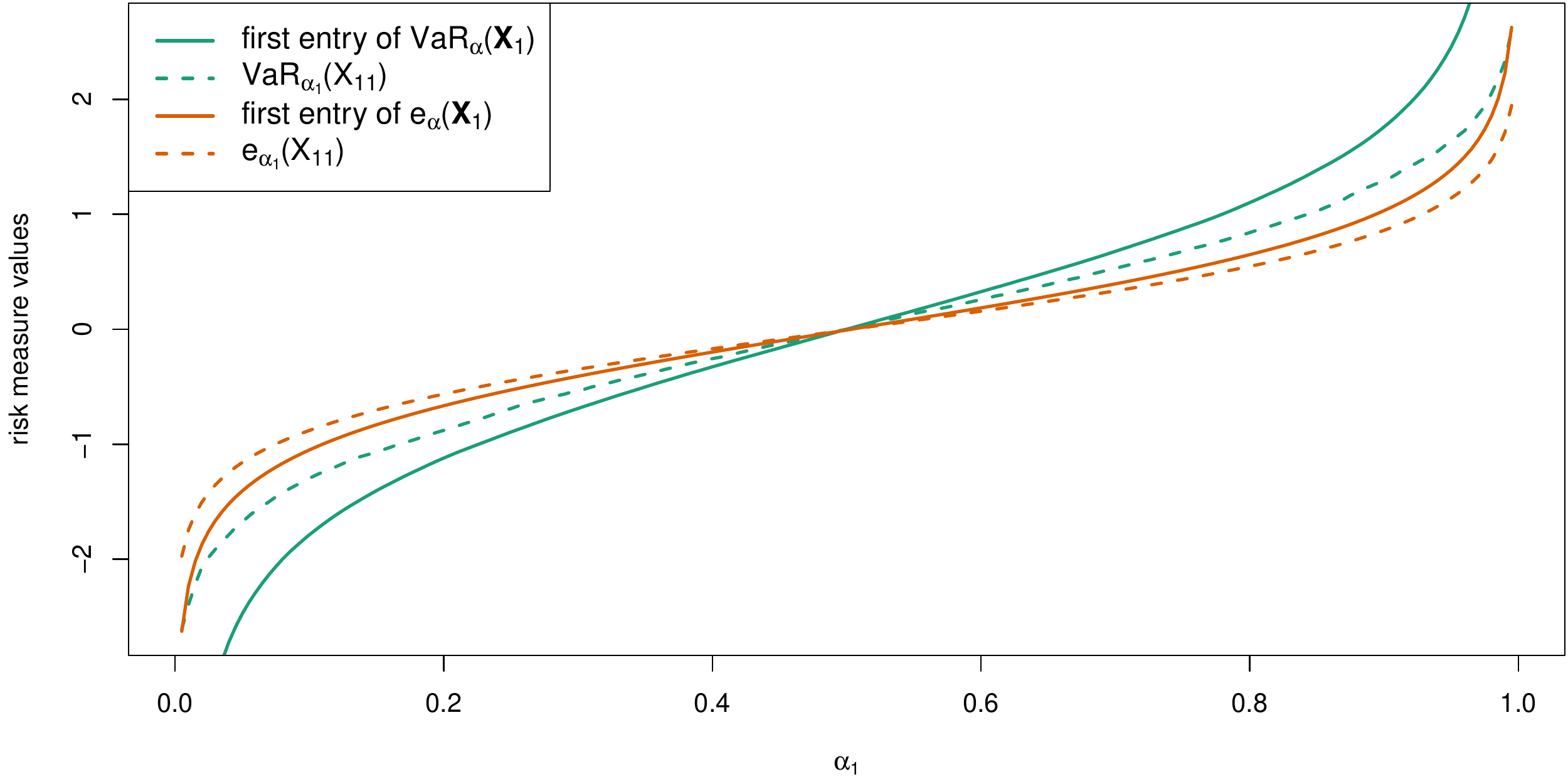}
  \caption{Solid lines indicate the first component of geometric $\VaR$ (green) and geometric expectiles (orange) based on $\bfalpha = (2\alpha_1-1)(1,0)$ applied to $\bfX_1 = (X_{11},X_{12})$, see Table~\ref{table_example_specification}. Dashed lines indicate univariate $\VaR$ (green) and expectiles (orange) at level $\alpha_1$ applied to $X_{11}$.}
\label{fig_compare_first_entry_normal}
\end{figure}
\begin{figure}
\centering
  \includegraphics[width=\textwidth]{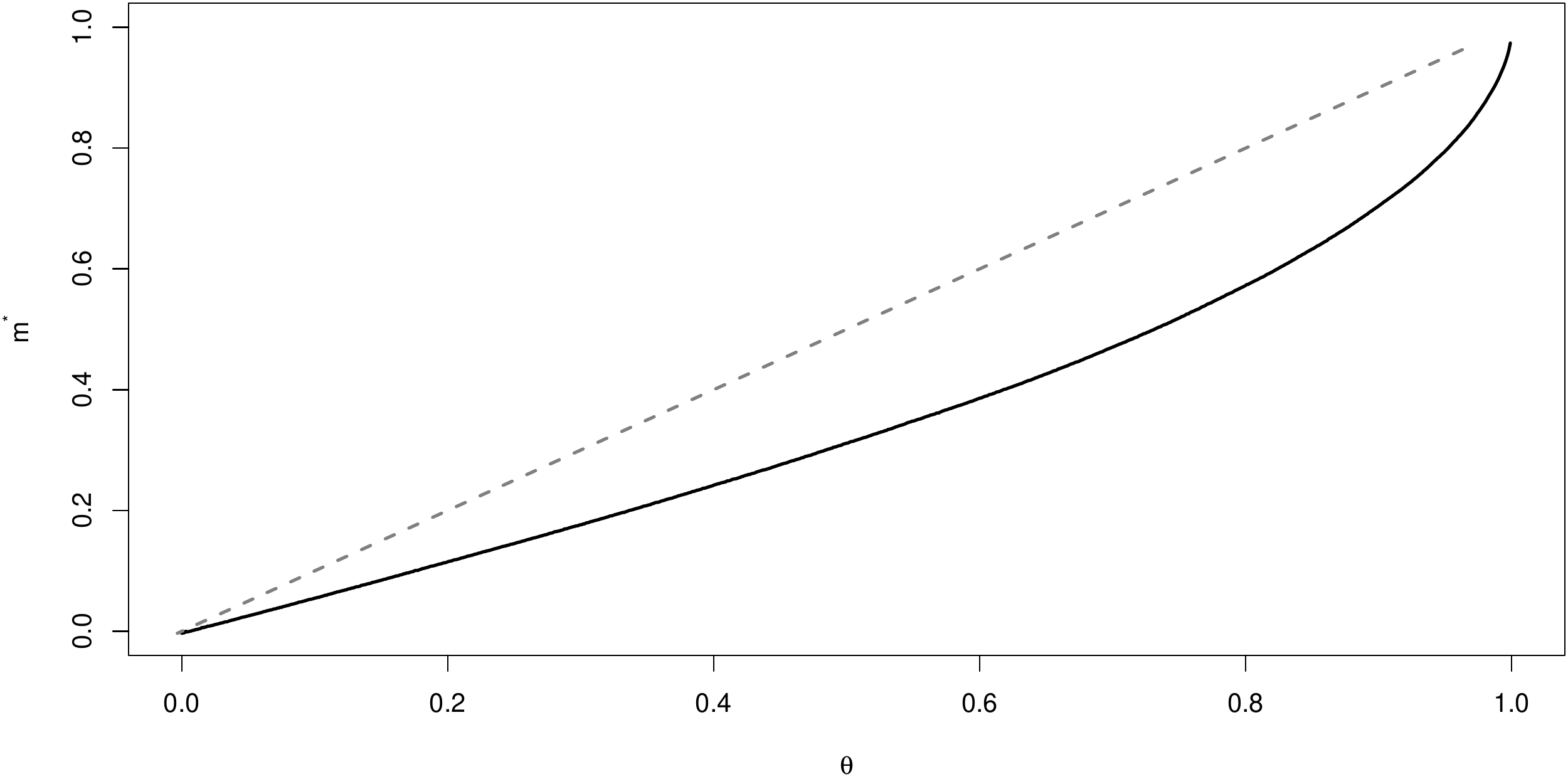}
	\caption{Pairs of magnitude $(\theta,m^*)$ such that $\VaR_{m^*\bfu}(\bfX_1) = \Expectile_{\bfalpha}(\bfX_1)$ in the least-square sense. Here $\bfu = (1,1)/\sqrt{2}$ and $\bfalpha = \theta \bfu$. $\bfX_1$ is specified in Table~\ref{table_example_specification}. The gray-dashed line indicates the $45$ degree line for comparison.}
\label{fig_equal_magnitude}
\end{figure}
\subsection{Higher Dimensional Marginalization}\label{section_higher_dim_margins}
While Section~\ref{sub_section_Comparing} has compared bivariate geometric expectiles to their univariate counterparts, it is of interest to compare geometric expectiles applied to higher dimensional margins to those applied to the full joint distribution.
Denote by $\bfX = (X_1,\ldots,X_d)$ a random vector of dimension $\dim$, and by $\bfY$ a sub-vector of $\bfX$ of dimensions $k < \dim$.
Without loss of generality we assume $\bfY = (X_1,\ldots,X_k)$.
Comparing $\Expectile_{\bfalpha}(\bfX)$ to $\Expectile_{\bfbeta}(\bfY)$ is challenging since the dimensions of the  respective indices $\bfalpha$ and $\bfbeta$ as well as the resulting vectors differ.
Disregarding the choice of indices for now it seems sensible to compare the first $k$ entries of $\Expectile_{\bfalpha}(\bfX)$ to $\Expectile_{\bfbeta}(\bfY)$.
This comparison would then focus on differences introduced by the dependence of $(X_1,\ldots,X_k)$ on $(X_{k+1},\ldots,X_d)$ which is neglected in $\Expectile_{\bfbeta}(\bfY)$.
Concerning the choice of indices $\bfalpha$ and $\bfbeta$, different scenarios are possible:
One possible choice is to first choose $\bfbeta \in B^{k}$ and then set $\bfalpha = (\bfbeta,0,\ldots,0)$.
In this case $\xNorm{\bfalpha}{2} = \xNorm{\bfbeta}{2}$ and $\bfalpha \in B^{\dim}$.
Alternatively, the vector $\bfalpha$ can be filled up with a vector $\bfz$ of non-zero values, that is $\bfalpha = (\bfbeta,\bfz)$.
In this case the condition $\xNorm{\bfalpha}{2} < 1$ needs to be obeyed whatever non-zero values are chosen, which immediately leads to $\xNorm{\bfalpha}{2} < 1$ if and only if $\xNorm{\bfz}{2} < \sqrt{1-\xNorm{\bfbeta}{2}^2}$.

To illustrate the effect of marginalization we consider the case $\dim = 3$ with $\bfX = (X_1,X_2,X_3)$ and $\bfY = (X_1,X_2)$.
We further set $\bfbeta_r(t) = r(\cos(t),\sin(t))^{\tr}$ where $0 < r < 1$ and $t \in [0,2\pi)$.
For $\bfalpha_r(t) = (\bfbeta_r(t),z(r))$ the possible values of $z(r)$ as a function of $r$ are then limited to the interval $(-\sqrt{1-r^2},\sqrt{1-r^2})$ to ensure $\xNorm{\bfalpha}{2} < 1$.

For the illustration the first marginal distribution $X_1$ of $\bfX$ follows a Gumbel distribution, $X_2 \sim t_4$ and $X_3$ follows a standard logistic distribution, while the dependence structure is given in terms of a Clayton copula with parameter $\theta=5$.
Consequently, $\bfY = (X_1,X_2)$ has the same Gumbel and $t_4$ margins also joined by a Clayton copula with parameter $\theta=5$.
In Figure~\ref{fig_2D_margins_expectile} we show the resulting geometric expectiles $\Expectile_{\bfbeta_r(t)}(\bfY)$ and the first two components of $\Expectile_{\bfalpha_r^i(t)}(\bfX)$, $i \in \{1,\ldots,7\}$, where $\bfalpha_r^i(t) = (\bfbeta_r(t),\ (-\tfrac{3}{4} + (i-1)\tfrac{1}{4})\sqrt{1-r^2})$.
Figure~\ref{fig_2D_margins_expectile} shows the results for $r=0.1$ (top left), $r=0.2$ (top right), $r=0.5$ (bottom left) and $r=0.9$ (bottom right).
From the figure we see that multiple intersections between the expectile curves $\Expectile_{\bfbeta_r(t)}(\bfY)$ and the first two components of $\Expectile_{\bfalpha_r^i(t)}(\bfX)$, $i \in \{1,\ldots,7\}$ are possible.
There is, however, one exception:
In case of $\bfalpha_r^4(t)$ we see that $\Expectile_{\bfbeta_r(t)}(\bfY)$ (orange) is always contained in the respective expectile curve based on $\bfalpha_r^4(t)$ (black).
For this choice of $\bfalpha$ the numerical result insinuates that the geometric expectiles of the sub-vector $\bfY$ are, as a set, contained in the respective components of the geometric expectiles of the full vector $\bfX$.
A partial explanation is that the components $(X_{k+1},\ldots,X_d)$ and their dependence with $(X_1,\ldots,X_k)$ are not at all taken into consideration when computing $\Expectile_{\bfbeta_r(t)}(\bfY)$.
While setting the respective elements in $\bfalpha^4$ to zero does eliminate the inner product terms associated with $(X_{k+1},\ldots,X_d)$ in \eqref{eq_multivariate_expectile}, see also \eqref{eq_multivariate_expectile_loss}, they are still contributing to the objective function via the norm term when computing $\Expectile_{\bfalpha_r^4(t)}(\bfX)$.
While this leads to comparatively wider spread contours, forcing $\xNorm{\bfalpha^4(t)}{2}=\xNorm{\bfbeta(t)}{2}$ continues to keep the results comparable.

In Figure~\ref{fig_2D_margins_VaR} we also compute geometric $\VaR_{\bfbeta_r(t)}(\bfY)$ and $\VaR_{\bfalpha_r^i(t)}(\bfX)$, $i \in \{1,\ldots,7\}$, for the same joint model $\bfX$ with $r=0.1$.
From the figure it is clear that geometric value-at-risk does not exhibit the ordering for indices $\bfbeta(t)$ and $\bfalpha_r^4(t)$ previously observed for geometric expectiles.
\begin{figure}
  \centering
  \includegraphics[width=0.49\textwidth]{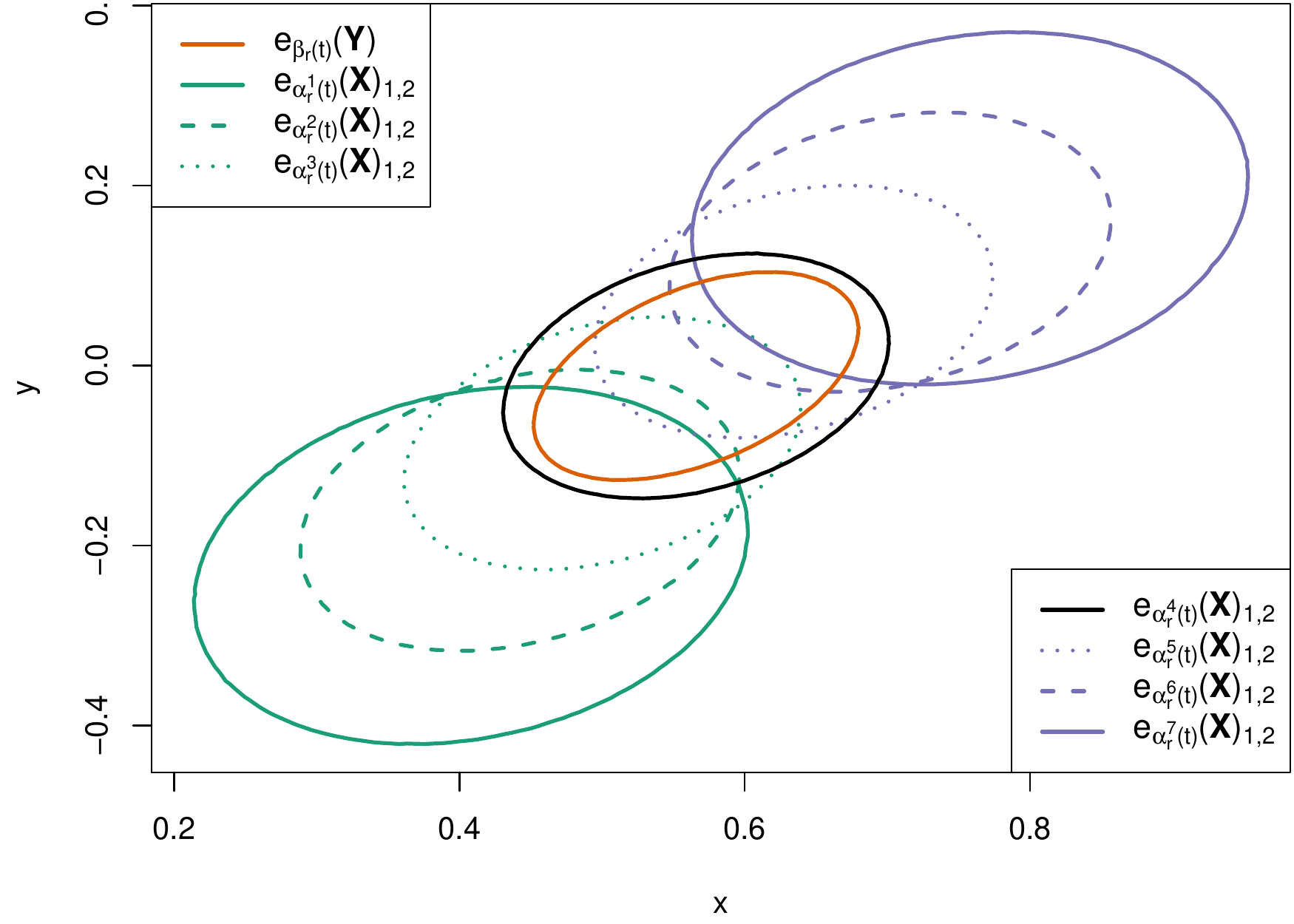}
  \includegraphics[width=0.49\textwidth]{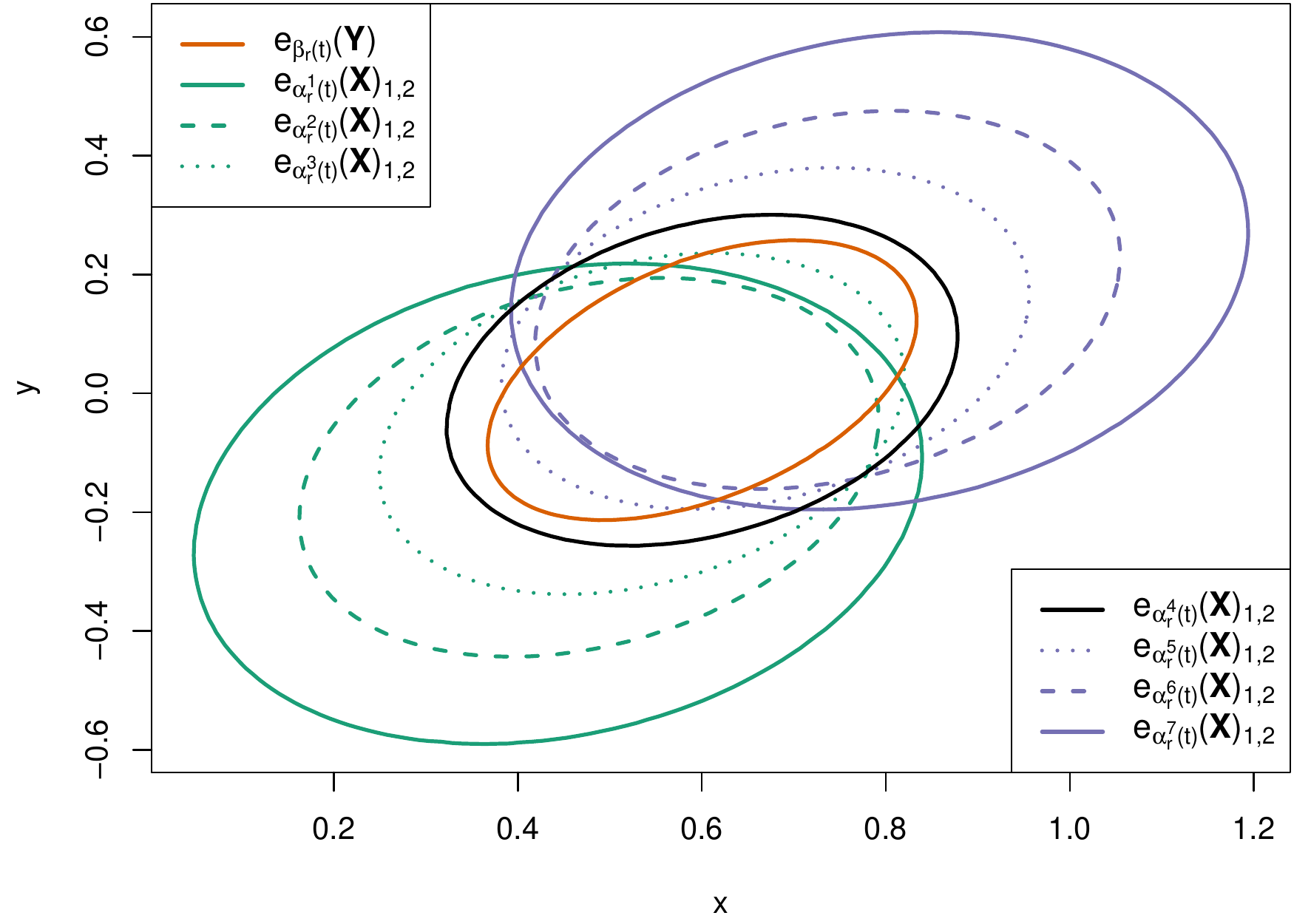}
  \includegraphics[width=0.49\textwidth]{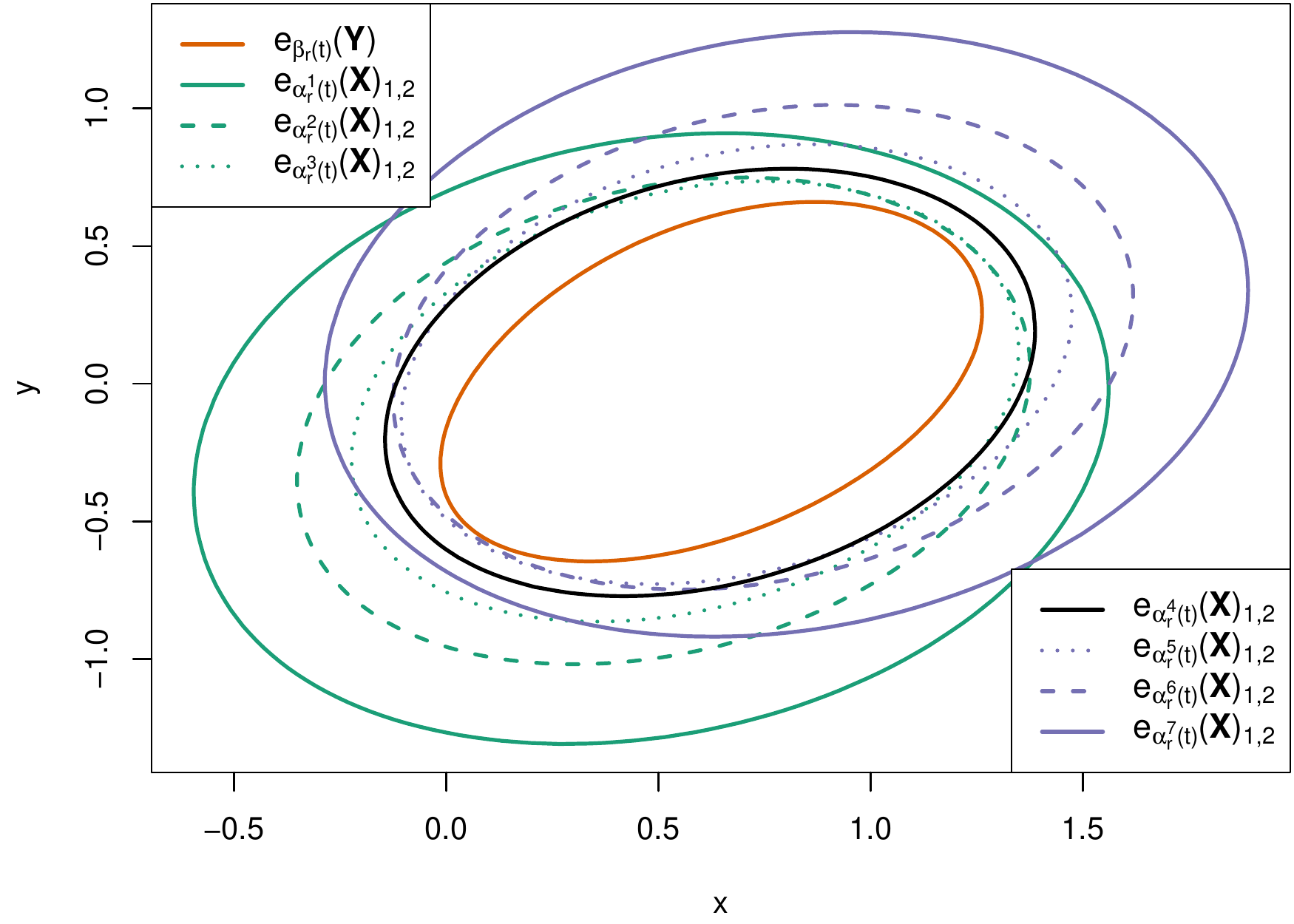}
  \includegraphics[width=0.49\textwidth]{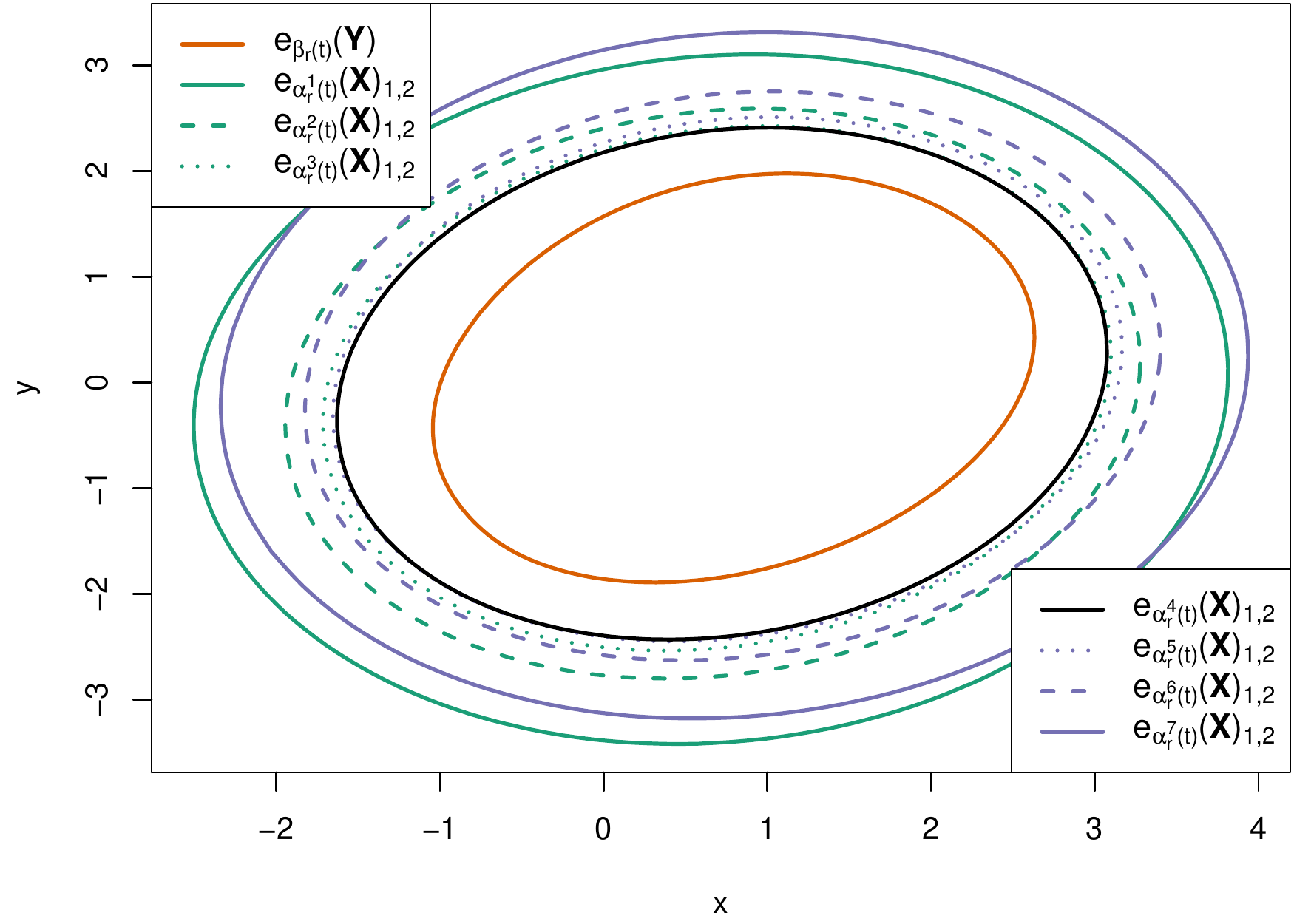}
	\caption{$\Expectile_{\bfbeta_r(t)}(\bfY)$ for $\bfbeta_r(t) = r(\cos(t),\sin(t))^{\tr}$, $t \in [0,2\pi)$ (orange). First two entries of $\Expectile_{\bfalpha^i_r(t)}(\bfX)$ for $i=1$ (green, solid), $i=2$ (green, dashed), $i=3$ (green, dotted), $i=4$ (black), $i=5$ (blue, dotted), $i=6$ (blue, dashed), $i=7$ (blue, solid). Radius $r=0.1$ (top left), $r=0.2$ (top right), $r=0.5$ (bottom left) and $r=0.9$ (bottom right). $\bfX=(X_1,X_2,X_3)$ and $\bfY=(X_1,X_2)$, where $X_1$ follows a Gumbel distribution, $X_2 \sim t_4$ and $X_3$ follows a standard logistic distribution. Dependence structure is given in terms of a Clayton copula with parameter $\theta=5$. Computations are based on $20000$ independent replications.}
    \label{fig_2D_margins_expectile}
\end{figure}
\begin{figure}
  \centering
  \includegraphics[width=\textwidth]{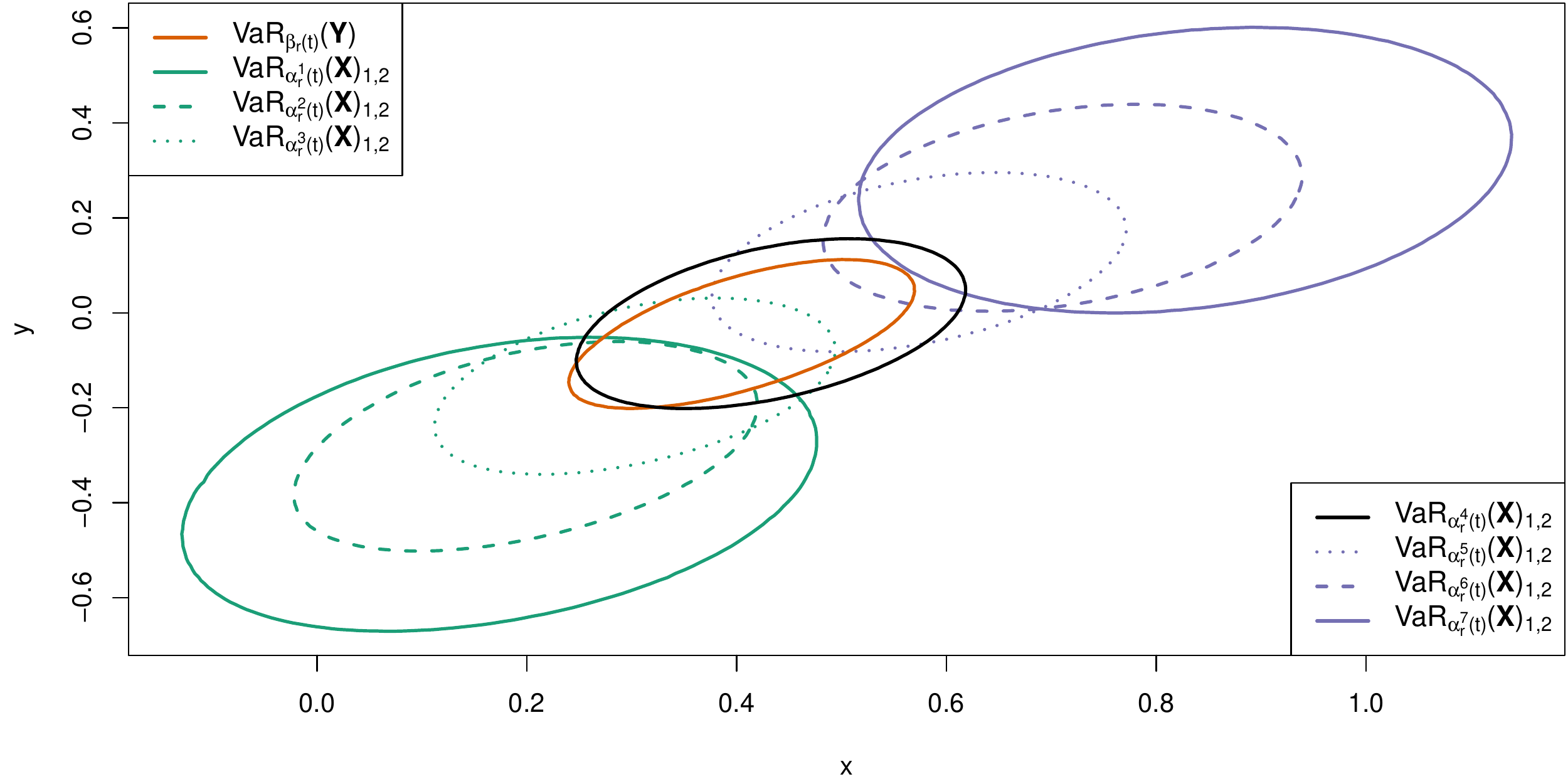}
	\caption{$\VaR_{\bfbeta_r(t)}(\bfY)$ for $\bfbeta_r(t) = r(\cos(t),\sin(t))^{\tr}$, $t \in [0,2\pi)$ (orange) and $r=0.1$. First two entries of $\VaR_{\bfalpha^i_r(t)}(\bfX)$ for $i=1$ (green, solid), $i=2$ (green, dashed), $i=3$ (green, dotted), $i=4$ (black), $i=5$ (blue, dotted), $i=6$ (blue, dashed), $i=7$ (blue, solid). $\bfX=(X_1,X_2,X_3)$ and $\bfY=(X_1,X_2)$, where $X_1$ follows a Gumbel distribution, $X_2 \sim t_4$ and $X_3$ follows a standard logistic distribution. Dependence structure is given in terms of a Clayton copula with parameter $\theta=5$. Computations are based on $20000$ independent replications.}
    \label{fig_2D_margins_VaR}
\end{figure}
\subsection{Bounded Random Vectors}
In this section we study the effect of applying geometric expectiles to a bounded random vector.
We therefore assume that $\bfX$ follows a Clayton copula $\Copula_{\theta}$ with parameter $\theta=5$, and compute $\Expectile_{\bfalpha(t)}(\bfX)$ for $\bfX \sim \Copula_5$ and $\bfalpha(t) = r(\cos(t),\sin(t))^{\tr}$ for $t \in [0,2\pi)$ and $r \in \{0.1,0.2,\ldots,0.9,0.95, 0.99, 0.9995, 0.9999, 0.99999\}$.
For extreme indices $\bfalpha$ the geometric expectile contours can be outside the support of $\bfX$ as shown in Figure~\ref{fig_bounded_support_expectile}.
Likewise, they can be outside of the convex hull of the data in an estimation setting.
This is in line with geometric $\VaR$, where $\xNorm{\VaR_{\bfalpha}(\bfX)}{2} \to \infty$ for sufficiently extreme indices $\xNorm{\bfalpha}{2} \to 1$, see \cite{GirardStupfler2017}.
To further study the behaviour when the norm of the underlying index tends to one we (numerically) study the function
\begin{align*}
d(r) = \xNorm{\Expectile_{\bfalpha(r)}(\bfX)-\E[\bfX]}{2},
\end{align*}
where $\bfalpha(r) = r\bfu$ for a fixed $\bfu$ with $\xNorm{\bfu}{2} = 1$ and $0 < r < 1$.
In Figure~\ref{fig_norm_function} we show an example of $d(r)$ for a four dimensional joint distribution when $\bfu = -(1,1,1,1)/\sqrt{4}$.
For the illustration the first marginal distribution $X_1$ follows a Gumbel distribution, $X_2 \sim t_4$, $X_3$ follows a standard logistic distribution and $X_4 \sim \N(0,1^2)$.
The dependence structure is given in terms of a Frank copula with parameter $\theta=3$.
Based on the numerical experiments it seems that $d(r)$ is monotonically increasing in $r$ with no limit in $\R$.
Aside from their related definitions this observation further supports the idea that geometric expectiles behave comparably to geometric $\VaR$ for extreme indices $\bfalpha$.
This potentially opens an avenue for studying the behaviour of geometric expectiles in a multivariate extreme value theory framework along the lines of \cite{GirardStupfler2015}.
\begin{figure}[!htbp]
  \centering
  \includegraphics[width=0.4\textwidth]{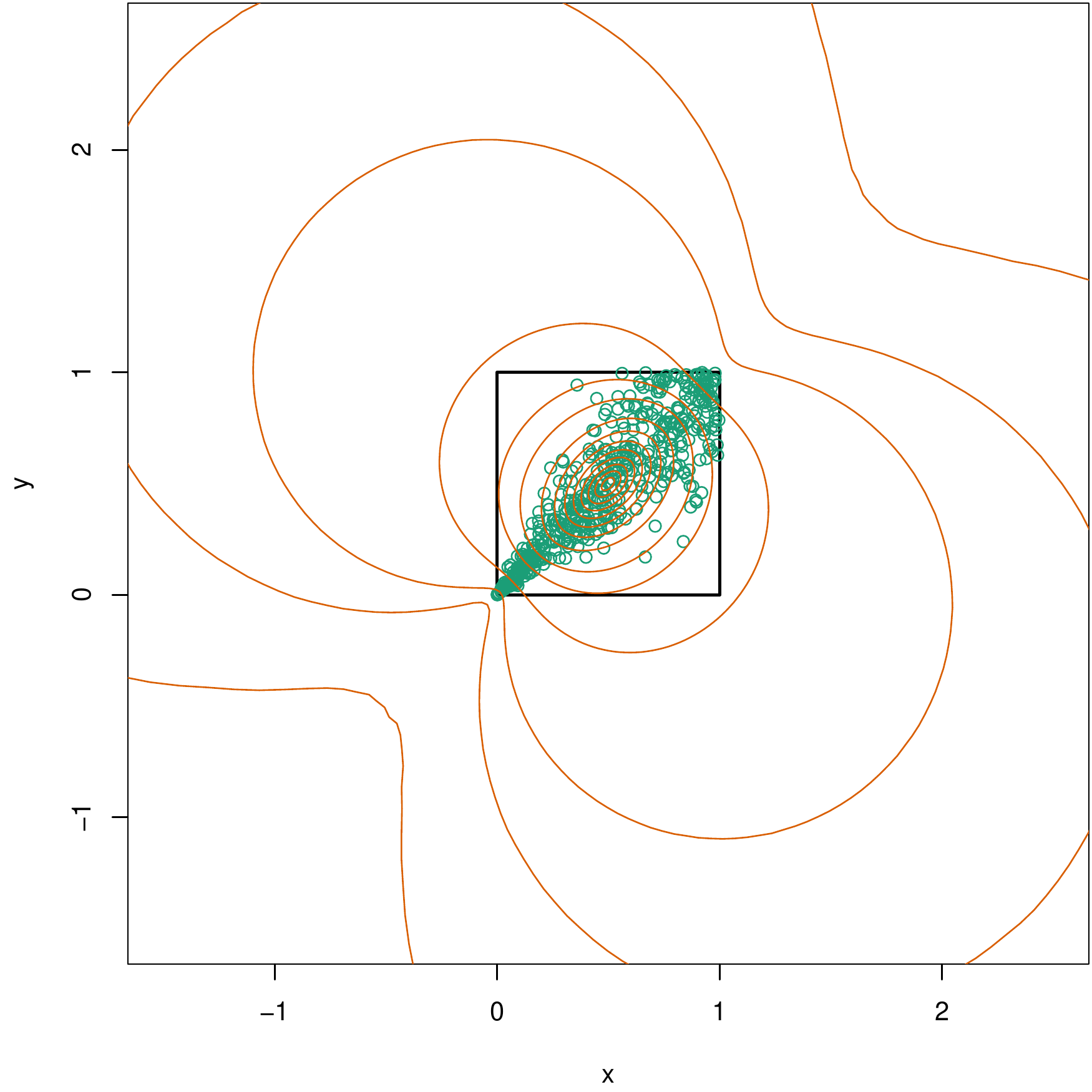}
	\caption{$\Expectile_{\bfalpha(t)}(\bfX)$ for $\bfX \sim \Copula_5$ and $\bfalpha(t) = r(\cos(t),\sin(t))^{\tr}$ for $r \in \{0.1,0.2,\ldots,0.9,0.95, 0.99, 0.9995, 0.9999, 0.99999\}$ and $t \in [0,2\pi)$ in orange. The black box indicates the support of $\bfX$, while green circles indicate a sample drawn from $\bfX$. Computations are based on $20000$ independent replications.}
    \label{fig_bounded_support_expectile}
\end{figure}
\begin{figure}
  \centering
\includegraphics[width=\textwidth]{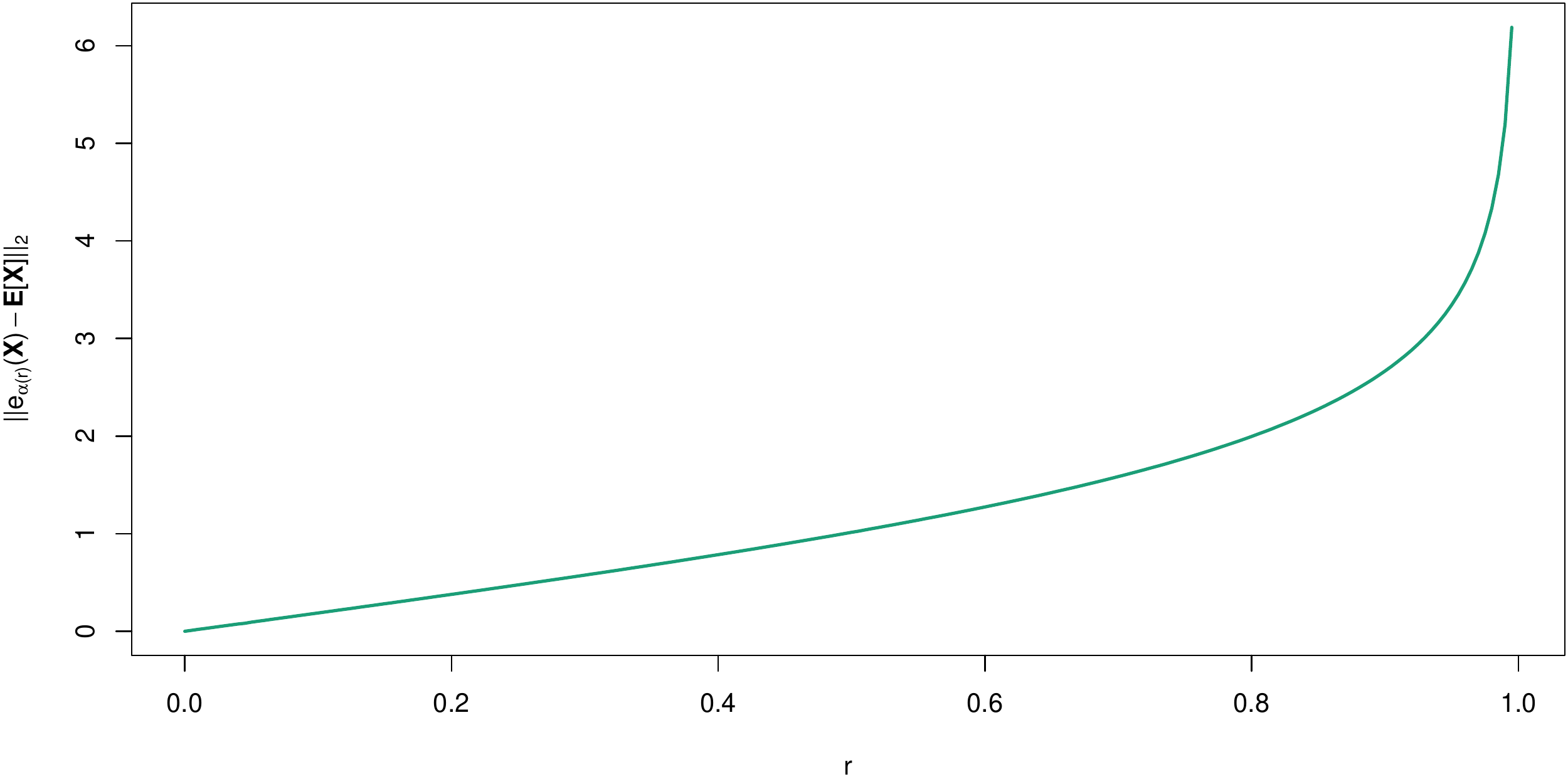}
\caption{$\xNorm{\Expectile_{\bfalpha(r)}(\bfX)-\E[\bfX]}{2}$ for $\bfX=(X_1,\ldots,X_4)$ when $X_1$ follows a Gumbel distribution, $X_2 \sim t_4$, $X_3$ follows a standard logistic distribution and $X_4 \sim \N(0,1^2)$. The dependence structure is given in terms of a Frank copula with parameter $\theta=3$. Furthermore, $\bfalpha(r) = -r(1,1,1,1)^{\tr}/\sqrt{4}$ for $r \in \{0, 0.005,0.01,0.015\ldots,0.995\}$. Computations are based on $20000$ independent replications.}
    \label{fig_norm_function}
\end{figure}
\subsection{Example Application}
To demonstrate how geometric expectiles can be used in a practical scenario we consider a data generating process that generalizes the well-known compound Poisson model.
By $\bfE = (E_1,E_2)$ we denote a random vector with exponentially distributed margins $E_\indI \sim \Exponential(\beta_\indI)$, $\indI = 1,2$.
The dependence structure between the components of $\bfE$ is given by a Clayton copula $\Copula_{\theta}$ with parameter $\theta > 0$.
For a Poisson random variable $N \sim \Poisson(\lambda)$ our final random vector $\bfX = (X_1,X_2)$ is then given by
\begin{align*}
\bfX = \sum^{N}_{k=1} \bfE_k,
\end{align*}
where $\bfE_k$ is an independent (of $N$ and $\bfE_j$ for $j \neq k$) copy of $\bfE$.
By construction we see that $X_j$, $j\in\{1,2\}$, is a compound Poisson model with exponentially distributed severities.
All in all the model captures the situation where a random number of risk occurs together, and the components of each incident are not independent.
Our example is motivated by considering vehicle insurance that
can\footnote{Coverage of medical costs depends on the respective jurisdiction. Vehicle insurance policies that cover medical and physical damage are common in the USA. On the other hand, there are, for example, no such products in the Qu\'{e}bec province of Canada since medical costs are in this case taken over by the province.}
cover medical payments for the insured party as well as physical damages to the insured vehicle.
From the point of the insurance company there will be a random number of accidents, where it is reasonable to assume a positive dependence between both components of the policy.

For our example we consider the parameters $\theta = 0.9$, $\beta_1 = 1/10$, $\beta_2 = 1/15$ and $\lambda = 1$.
The computation of the geometric expectiles is now based on a simulated iid sample $(\bfx_i)_{i=1}^{100}$ of $\bfX$.
The computation therefore utilizes the Monte Carlo estimator according to Corollary \ref{corolloary_consistency} and the discussion therein.
Figure~\ref{fig_numerical_expectiles_compound_model} shows the resulting geometric expectiles, where we again consider the previously introduced, see Figure~\ref{fig_possible_directions} and Section~\ref{section_numerical_illustration}, indices $\bfalpha_1(\varphi) = 0.98(\cos(\varphi),\sin(\varphi))$ and $\bfalpha_2(\varphi) = (0.98\cos(\varphi),0.90\sin(\varphi))$.
Given that in this example the margins are a.s.\ positive, we confine ourselves to directions in the first quadrant only, i.e., $\varphi \in [ 0,\pi /2 ]$.
In this case numbers indicate the resulting geometric expectiles for indices $\bfalpha_j(\varphi_k)$, $j\in\{1,2\}$, where $\varphi_k = k \pi/14 $, $k\in\{0,\ldots,7\}$.

Concerning the individual variables $X_1$ and $X_2$ the insurer can now reserve losses according to the indices $\bfalpha_j(\varphi_0)$ or respectively $\bfalpha_j(\varphi_7)$.
Taking $j=1$ corresponds to a traditional confidence level of $0.99$ for both components, while $j=2$ corresponds to a traditional confidence level of $0.99$ for $X_1$ and $0.95$ for $X_2$.
More importantly, by extending the univariate forecast model validation theory outlined, for example, in \cite{Gneiting2011} or \cite{NoldeZiegel2017} it might be possible to validate the proposed model against real data by backtesting.
Using geometric expectiles the backtest would then validate the full joint distribution function of $\bfX$, and not just the individual marginal distributions of $X_1$ and $X_2$.
\begin{figure}[!htbp]
\centering
  \includegraphics[width=\textwidth]{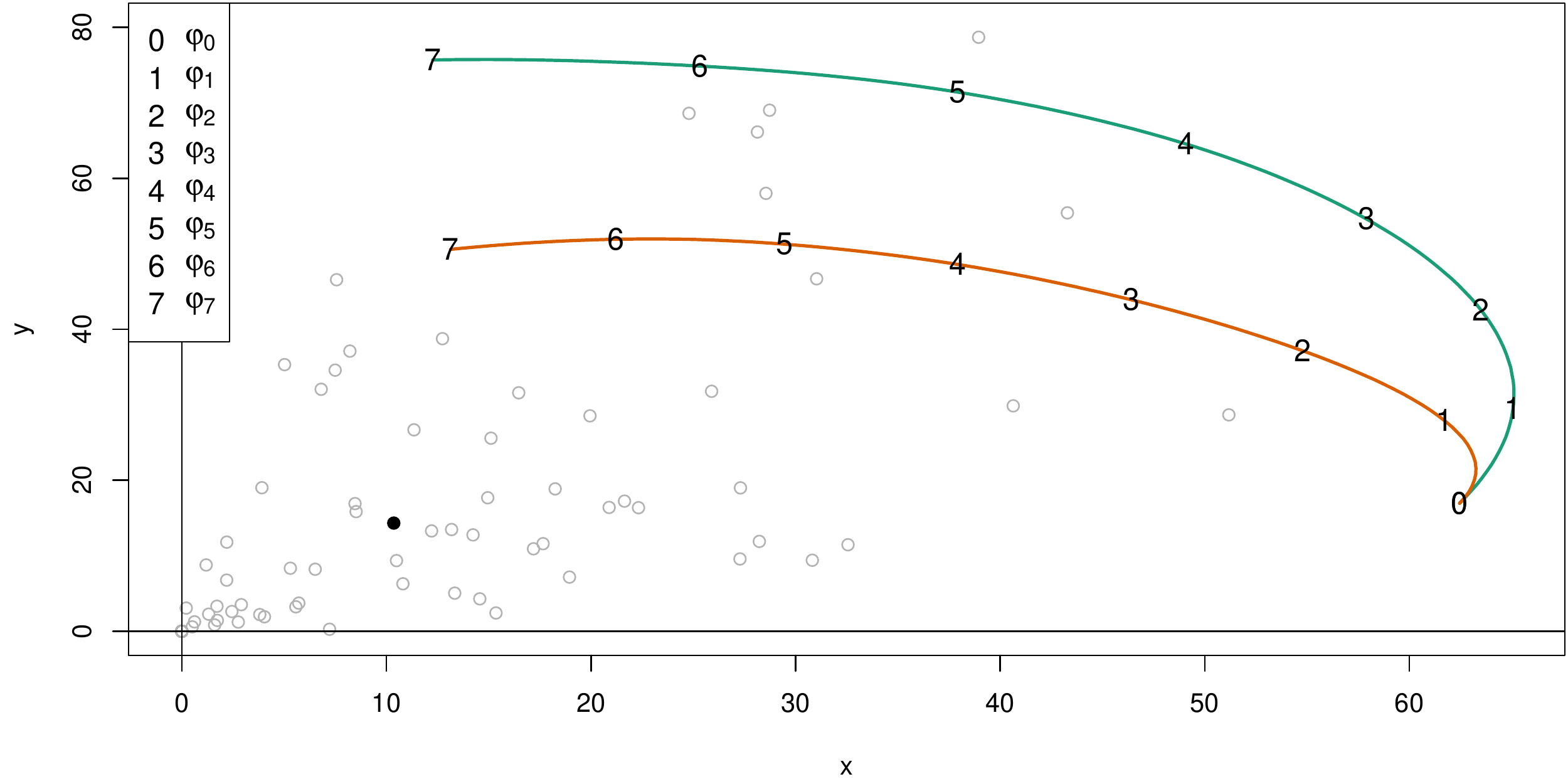}
	\caption{Geometric expectiles $\Expectile_{\bfalpha_1}(\bfX)$ (green line) and $\Expectile_{\bfalpha_2}(\bfX)$ (orange line) for the bivariate compound Poisson model. Indices are given as $\bfalpha_1(\varphi) = 0.98(\cos(\varphi),\sin(\varphi))$ and $\bfalpha_2(\varphi) = (0.98\cos(\varphi),0.90\sin(\varphi))$, where $\varphi \in [ 0,\pi/2 ]$. Numbers indicate the resulting geometric expectiles for indices $\bfalpha_j(\varphi_k)$, $j\in\{1,2\}$, where $\varphi_k = k \pi/14 $, $k\in\{0,\ldots,7\}$. The black dot indicates the bivariate mean $(\E[X_1],\E[X_2])$. Computations are based on $100$ iid realizations of $\bfX$ marked by gray circles.}
\label{fig_numerical_expectiles_compound_model}
\end{figure}
\section{Conclusion}\label{section_conclusion}
In this paper we introduced geometric expectiles for multivariate distribution functions with finite second moments of the margins.
This proposed functional naturally generalizes univariate expectiles introduced in \cite{NeweyPowell1987} to the multivariate case for any fixed dimension $\dim$.
Instead of a single real number, geometric expectiles are represented by a $\dim$-dimensional vector, which can be used for risk management purposes, for risk selection and comparison.
This approach is in line with other recently introduced multivariate risk measures.
Utilizing a framework comparable to the one introduced in \cite{Chaudhuri1996} to generalize quantiles, the resulting geometric expectiles are indexed by an element of the open unit ball of $\R^{\dim}$.

Seen as a statistical functional, geometric expectiles have a number of desirable properties.
First, they are well-defined and unique for any multivariate distribution function with margins with finite second moments.
Second, multivariate geometric expectiles have desirable properties under data transformations such as translating, re-scaling or re-ordering the data.
Generalizing a re-ordering, also multiplications with orthogonal matrices are well behaved.
Third, as in the univariate case, geometric expectiles are elicitable in a multivariate sense.
Comparable to the univariate case, this may provide one with a mechanism to rank competing multivariate forecasting procedures, or to backtest a multivariate model against real data.

Aside from population characteristics, we also studied properties and asymptotics of the corresponding finite sample version.
Here we find that the sample version is a consistent estimator of the population characteristics.
A Monte Carlo estimator of geometric expectiles is readily available when a closed-form solution is not.
Furthermore, to reduce the variance of the numerical estimates, quasi-Monte Carlo methods can be employed to improve the variance of the Monte Carlo estimators of the expectations in \eqref{eq_multivariate_var} and \eqref{eq_multivariate_expectile}.
This simplifies the computation of the minimizer from a numerical point of view.

In the presented examples, we utilized these simulation-based approximations to contrast geometric expectiles to the geometric quantiles introduced in \cite{Chaudhuri1996} as well as univariate expectiles and quantiles.
Our results indicate that geometric value-at-risk is more conservative than geometric expectiles for a given index.

In cases where the second moment condition on the margins is too restrictive it remains to be seen how tempering the margins interacts with geometric expectiles, providing a possible remedy.

Despite the extent of the present study, we can identify the following open questions concerning multivariate geometric expectiles:
It is unclear which stochastic order $\prec$ between random vectors is compatible with the corresponding geometric expectiles, so that $\Expectile_{\bfalpha}(\bfX) \sqsubset \Expectile_{\bfalpha}(\bfY)$ if $\bfX \prec \bfY$ in this order.
Furthermore, while the multivariate generalization of subadditivity proposed in this paper can numerically be verified for a wide range of distributions, it remains unclear how this property can be shown analytically.
The same holds true for the marginalization discussed in Section~\ref{section_higher_dim_margins} where we observed numerically the ordering of geometric expectiles when applied to (higher dimensional) margins and the full distribution.
Concerning the distance of $\Expectile_{\bfalpha}(\bfX)$ to $\E[\bfX]$, our findings are in line with geometric $\VaR$ and thus it is reasonable to expect a monotonic divergence to $\infty$.
In the special case of bounded random vectors this may hamper a straightforward application of geometric expectiles as risk measures, and addressing this issue will be part of further research.

Finally, while it is known, see \cite{Koltchinskii1997}, that geometric $\VaR_{\bfalpha}(\bfX)$ fully characterizes the joint distribution of $\bfX$, it is not clear if this also holds for geometric $\Expectile_{\bfalpha}(\bfX)$.

\subsection*{Acknowledgements}
This work was supported by NSERC under Grant RGPIN-5010-2015 and RGPIN-2015-05447. The authors would also like to thank Connor Jackman for communicating a vital step in the proof of Theorem~\ref{theorem_parallelogram_inequality}.

\bibliographystyle{plainnat}
\bibliography{proposition}
\end{document}